\documentclass[11pt]{article}
\usepackage[margin=1in]{geometry}
\usepackage{amsmath}
\usepackage{appendix}
\usepackage{array}
\usepackage{tabularx}  
\usepackage{amsmath,amsthm,amssymb,amsfonts}
\usepackage[american]{babel}
\usepackage[url=false,giveninits=true,maxbibnames=99,style=alphabetic,maxalphanames=5]{biblatex}
\usepackage{mathtools}
\usepackage{enumitem}
\usepackage{csquotes}
\usepackage[noend]{algpseudocode}
\usepackage{algorithm}
\usepackage{xparse}
\usepackage{xspace}
\usepackage{color}
\usepackage{caption}
\usepackage{subcaption}
\usepackage{fullpage}
\usepackage{placeins}
\usepackage{xcolor}
\usepackage{makecell}
\usepackage{qcircuit}
\usepackage{thmtools}
\usepackage{thm-restate}
\usepackage[textsize=scriptsize]{todonotes}
\usepackage{verbatim}
\usepackage{qcircuit}
\usepackage{comment}
\usepackage{mathdots}
\usepackage{hyperref}
\usepackage{cleveref}

\setuptodonotes{color=blue!15}

\newcommand{\Luke}{\color{blue}}

\mathchardef\mhyphen="2D 

\newcommand\newmathabbrev[2]{\newcommand{#1}{\ensuremath{#2}\xspace}}

\newcommand\cfont\mathsf
\newmathabbrev\p{\cfont{P}}
\newmathabbrev{\N}{\mathbb N}
\newmathabbrev\NP{\cfont{NP}}
\newmathabbrev\QPH{\cfont{QPH}}
\newmathabbrev\QPHpure{\cfont{pureQPH}}
\newmathabbrev\QCPH{\cfont{QCPH}}
\newmathabbrev\polyQCPH{\cfont{polyQCPH}}
\newmathabbrev\polyPH{\cfont{polyPH}}
\newmathabbrev\QEPH{\cfont{QEPH}}
\newmathabbrev\QMAH{\cfont{QMAH}}
\newmathabbrev\QAC{\cfont{QAC}_0}
\newmathabbrev\AC{\cfont{AC}}
\newmathabbrev\disagr{\mathsf{disagr}}

\newmathabbrev\SIP{\cfont{SIPSER}}

\newmathabbrev\DTIME{\cfont{DTIME}}
\newmathabbrev\tSAT{3\cfont{\mhyphen{}SAT}}
\newmathabbrev\MA{\cfont{MA}}
\newmathabbrev\AM{\cfont{AM}}
\newmathabbrev\NPDAG{\cfont{NP\mhyphen{}DAG}}
\newmathabbrev\QMADAG{\cfont{QMA\mhyphen{}DAG}}
\newmathabbrev\yes{\mathrm{yes}}
\newmathabbrev\no{\mathrm{no}}
\newmathabbrev\US{\cfont{US}}
\newmathabbrev\FP{\cfont{FP}}
\newmathabbrev\PP{\cfont{PP}}
\newmathabbrev\CeP{\cfont{C_=P}}
\newmathabbrev\coCeP{\cfont{coC_=P}}
\newmathabbrev\PH{\cfont{PH}}
\newmathabbrev\SAT{\cfont{SAT}}
\newmathabbrev\SPP{\cfont{SPP}}
\newmathabbrev\GapP{\cfont{GapP}}
\newmathabbrev\BQP{\cfont{BQP}}
\newmathabbrev\ZQEXP{\cfont{ZQEXP}}
\newmathabbrev\QP{\cfont{QP}}
\newmathabbrev\StoqMA{\cfont{StoqMA}}
\newmathabbrev\coNP{\cfont{coNP}}
\newmathabbrev\AzPP{\cfont{A_0PP}}
\newmathabbrev\QMA{\cfont{QMA}}
\newmathabbrev\QXC{\cfont{QXC}}
\newmathabbrev\QMAone{\cfont{QMA}_1}
\newmathabbrev\uQMA{\cfont{uniqueQMA}}
\newmathabbrev\uSAT{\cfont{uniqueSAT}}
\newmathabbrev\cQMA{\cfont{cloneableQMA}}
\newmathabbrev\coQMA{\cfont{coQMA}}
\newmathabbrev\BPP{\cfont{BPP}}
\newmathabbrev\QCMA{\cfont{QCMA}}
\newmathabbrev\pNPlog{\p^{\NP[\log]}}
\newmathabbrev\pNP{\p^{\NP}}
\newmathabbrev\pNPtwo{\p^{\NP[2]}}
\newmathabbrev\pNPone{\p^{\NP[1]}}
\newmathabbrev\pParSAT{\p^{||\SAT}}
\newmathabbrev\pQMApar{\p^{||\QMA}}
\newmathabbrev\pCpar{\p^{||\C}}
\newmathabbrev\pStoqMApar{\p^{||\StoqMA}}
\newmathabbrev\pQMAlog{\p^{\QMA[\log]}}
\newmathabbrev\pClog{\p^{\textup{C}[\log]}}
\newmathabbrev\pC{\p^{\textup{C}}}
\newmathabbrev\QMASPACE{\cfont{QMASPACE}}
\newmathabbrev\pQMAtlog{\p^{\QMA(2)[\log]}}
\newmathabbrev\pStoqMAlog{\p^{\StoqMA[\log]}}
\newmathabbrev\pQMApt{\p^{\Vert\QMA(2)}}
\newmathabbrev\pQMA{\p^{\QMA}}
\newmathabbrev\SharpP{\cfont{\#P}}
\newmathabbrev\pSharP{\p^{\SharpP[1]}}
\newmathabbrev\PromisePP{\cfont{PromisePP}}
\newmathabbrev\lett{\le_\mathrm{tt}}
\newmathabbrev\YES{\mathsf{YES}}
\newmathabbrev\NO{\mathsf{NO}}
\newmathabbrev\PSPACE{\cfont{PSPACE}}
\newmathabbrev\IP{\cfont{IP}}
\newmathabbrev\POLY{\cfont{POLY}}
\newmathabbrev\DAG{\cfont{DAG}}
\newmathabbrev\StoqMADAG{\StoqMA\mhyphen\cfont{DAG}}
\newmathabbrev\CDAG{C\mhyphen\cfont{DAG}}
\newmathabbrev\CDAGf{C\mhyphen\cfont{DAG}_f}
\newmathabbrev\CDAGs{C\mhyphen\cfont{DAG}_s}
\newmathabbrev\CDAGd{C\mhyphen\cfont{DAG}_{d}}
\newmathabbrev\CDAGo{C\mhyphen\cfont{DAG}_1}
\newmathabbrev\LOGS{\cfont{LOGS}}
\newmathabbrev\TAUT{\cfont{TAUTOLOGY}}
\newmathabbrev\SBQP{\cfont{SBQP}}
\newmathabbrev\Fc{F_\coNP}
\newmathabbrev\Fa{F_\AzPP}
\newmathabbrev\GSCON{\cfont{GSCON}}
\newmathabbrev\GSCONexp{\GSCON_\cfont{exp}}
\newmathabbrev\QMAexp{\QMA_\cfont{exp}}
\newmathabbrev\UQMA{\cfont{UQMA}}
\newmathabbrev\R{\mathbb R}
\newmathabbrev\Trees{\cfont{TREES}}
\newmathabbrev\apxsim{\cfont{APX\mhyphen{}SIM}}
\newmathabbrev\AWPP{\cfont{AWPP}}
\newmathabbrev\X{\mathcal{X}}
\newmathabbrev\Y{\mathcal{Y}}

\newmathabbrev\Z{\mathcal{Z}}

\newmathabbrev\ZZ{\mathbb{Z}}
\newmathabbrev\Hprop{H_\mathrm{prop}}
\newmathabbrev\Hin{H_\mathrm{in}}
\newmathabbrev\Hout{H_\mathrm{out}}
\newmathabbrev\Hstab{H_\mathrm{stab}}
\newmathabbrev\Lext{\L_\mathrm{ext}}
\newmathabbrev\BTWNP{\cfont{BTW}(\NP)}
\newmathabbrev\BSN{\cfont{BSN}}
\newmathabbrev\SN{\cfont{SN}}
\newmathabbrev\BD{\cfont{BD}}
\newmathabbrev\HYPERTREE{\cfont{NP\mhyphen{}HYPERTREE}}
\newmathabbrev\Hext{H_\mathrm{ext}}
\newmathabbrev\Hpropt{\tilde{H}_\mathrm{prop}}
\newmathabbrev\Hint{\tilde{H}_\mathrm{in}}
\newmathabbrev\Houtt{\tilde H_\mathrm{out}}
\newmathabbrev\EXP{\cfont{EXP}}
\newmathabbrev\A{\mathcal{A}}
\newmathabbrev\U{\mathcal{U}}

\renewcommand\L{\mathcal{L}}

\newmathabbrev\DAGSSAT{\DAGS(\SAT)}
\newmathabbrev\DAGS{\mathrm{DAGS}}
\newmathabbrev\DAGSNP{\DAGS(\NP)}

\newmathabbrev\AND{\cfont{AND}}

\newmathabbrev\STCONN{{S,T}\cfont{\mhyphen{}CONN}}
\newmathabbrev\CNF{\cfont{CNF}}
\newmathabbrev\NEXP{\cfont{NEXP}}
\newmathabbrev\NPSPACE{\cfont{NPSPACE}}
\newmathabbrev\QCMASPACE{\cfont{QCMASPACE}}
\newmathabbrev\BQPSPACE{\cfont{BQPSPACE}}
\newmathabbrev\PQPSPACE{\cfont{PQPSPACE}}
\newmathabbrev\PQP{\cfont{PQP}}
\newmathabbrev\TQBF{\cfont{TQBF}}
\newmathabbrev{\PCP}{\cfont{PCP}}
\newmathabbrev\BQUPSPACE{\cfont{BQ_UPSPACE}}
\newmathabbrev\QMAt{\QMA(2)}
\newmathabbrev\QMAtSEP{\QMA^{\mathsf{SEP}}(2)}
\newmathabbrev\QMAtexp{\QMAt_{\exp}}
\newmathabbrev\MIP{\cfont{MIP}}
\newmathabbrev\MIPt{\MIP(2)}
\newmathabbrev\BellQMA{\cfont{BellQMA}}
\newmathabbrev\BellQMAt{\BellQMA(2)}
\newmathabbrev\BellQMAtexp{\BellQMAt_{\exp}}

\protected\def\verythinspace{%
  \ifmmode
    \mskip0.5\thinmuskip
  \else
    \ifhmode
      \kern0.08334em
    \fi
  \fi
}

\newcommand{\C}{\mathbb C}

\newcommand{\be}{\begin{equation}}
\newcommand{\ee}{\end{equation}}

\renewcommand{\epsilon}{\varepsilon}

\DeclareMathOperator{\Tr}{Tr}

\DeclarePairedDelimiter\bra{\langle}{\rvert}
\DeclarePairedDelimiter\ket{\lvert}{\rangle}

\DeclarePairedDelimiterX\braket[2]{\langle}{\rangle}{#1 \delimsize\vert #2}
\DeclarePairedDelimiterX\ketbra[2]{\lvert}{\rvert}{#1 \delimsize\rangle\delimsize\langle #2}

\setlist[itemize]{noitemsep, topsep=0pt}
\setlist[enumerate]{noitemsep, topsep=0pt}

\declaretheorem[numberwithin=section]{theorem}

\declaretheorem[sibling=theorem]{corollary}
\declaretheorem[sibling=theorem]{lemma}

\declaretheorem[sibling=theorem]{fact}

\crefname{observation}{observation}{observations}
\Crefname{observation}{Observation}{Observations}

\makeatletter
\newcommand{\subalign}[1]{%
  \vcenter{%
    \Let@ \restore@math@cr \default@tag
    \baselineskip\fontdimen10 \scriptfont\tw@
    \advance\baselineskip\fontdimen12 \scriptfont\tw@
    \lineskip\thr@@\fontdimen8 \scriptfont\thr@@
    \lineskiplimit\lineskip
    \ialign{\hfil$\m@th\scriptstyle##$&$\m@th\scriptstyle{}##$\hfil\crcr
      #1\crcr
    }%
  }%
}
\NewDocumentCommand{\LeftComment}{s m}{%
  \Statex \IfBooleanF{#1}{\hspace*{\ALG@thistlm}}\(\triangleright\) #2}

\def\moverlay{\mathpalette\mov@rlay}
\def\mov@rlay#1#2{\leavevmode\vtop{%
   \baselineskip\z@skip \lineskiplimit-\maxdimen
   \ialign{\hfil$\m@th#1##$\hfil\cr#2\crcr}}}
\newcommand{\charfusion}[3][\mathord]{
    #1{\ifx#1\mathop\vphantom{#2}\fi
        \mathpalette\mov@rlay{#2\cr#3}
      }
    \ifx#1\mathop\expandafter\displaylimits\fi}

\makeatother

\algnewcommand{\LineComment}[1]{\State \(\triangleright\) #1}

\algblockdefx[ON]{Blk}{EndBlk}[1]
  {#1}
  {}

\makeatletter
\ifthenelse{\equal{\ALG@noend}{t}}%
  {\algtext*{EndBlk}}
  {}%
\makeatother

\AtEveryBibitem{%
  \clearlist{language}%
}

\def\({\left(}
\def\){\right)}
\def\X{\mathcal{X}}
\def\Y{\mathcal{Y}}
\def\Z{\mathcal{Z}}

\def\yes{\text{yes}}
\def\no{\text{no}}

\usepackage{authblk}
\usepackage{longtable} 
\usepackage{booktabs}  
\usepackage{seqsplit}
\usepackage{comment}

\usepackage{ytableau}

\makeatletter
\renewcommand\@makefnmark{\hbox{\textsuperscript{\@thefnmark}}}
\makeatother

\usepackage{amssymb}
\usepackage{tikz}
\usetikzlibrary{calc}
\usetikzlibrary{positioning}
\usepackage{hyperref}
\usepackage{xurl}

\hypersetup{
colorlinks = true,
citecolor= blue,
linkcolor= blue,
breaklinks=true
}
\title{\textsf{Enhanced quantum capacity thresholds from symmetry}}

\date{}

\author[1,2]{Avantika Agarwal}
\author[1,2,3]{Amolak Ratan Kalra}
\author[1,4]{Sungjai Lee}
\author[1,3,5]{Debbie Leung}
\author[1,2]{\\Luke Schaeffer}
\author[1,2]{Pulkit Sinha}
\author[1,4]{Graeme Smith}
\affil[1]{Institute for Quantum Computing, University of Waterloo}
\affil[2]{David R. Cheriton School of Computer Science, University of Waterloo}
\affil[3]{Perimeter Institute for Theoretical Physics}
\affil[4]{Department of Applied Mathematics, University of Waterloo}
\affil[5]{Department of Combinatorics and Optimization, University of Waterloo}

\begin{document}

\maketitle

\begin{abstract}

The quantum capacity captures the value of a quantum channel for transmitting quantum information, establishing the fundamental limits on quantum communication. In spite of its central role in quantum information theory, the quantum capacity of most channels is unknown, with wide gaps between the best upper and lower bounds.  Even deciding whether a channel has nonzero capacity---finding its capacity threshold---is difficult.  In this paper we report significant increases in the capacity thresholds of two prototypical noise models: the depolarizing channel and Pauli channels.  In the case of the depolarizing channel, this is the first improvement in 18 years, giving a bigger increase beyond the hashing bound than all previous improvements combined.  Our starting point is the representation theoretic framework recently proposed by Bhalerao and Leditzky (2025) to compute coherent information for special permutation invariant states.  We generalize their framework to the full symmetric subspace, which allow us to optimize coherent information over rank two states in that space. A representation theoretic calculation shows that exponentially many Kraus operators of the channel annihilate the symmetric space, corresponding to a massive decrease in environment entropy for states on the symmetric space compared to the maximally mixed state. This explains the enhanced coherent information as a manifestation of degeneracy for the resulting codes.

\end{abstract}
\section{Introduction}
\label{sec:introduction}
The quantum capacity of a noisy quantum channel is the maximum rate at which reliable quantum communication is possible over the channel. Measured in qubits per channel use, it captures the optimal performance of a quantum error correcting code adapted to the channel's noise model. For the most part, we don't yet have the tools necessary to evaluate quantum capacity.  For example, while the qubit depolarizing channel is one of the simplest and most natural quantum channels, evaluating its quantum capacity is a major open problem in quantum information theory. 

In the classical setting the analogous question is well understood. Shannon \cite{shannon} in 1948 showed that the classical capacity of any channel is equal to the mutual information between the input and output when maximized over the input distribution. A great simplifying feature of the classical problem is that the maximum mutual information between input and output when optimized over $n$ copies of the channel is equal to $n$ times the optimal value for a single copy of the channel. This property is called \emph{additivity} of the maximum mutual information for a channel.

Surprisingly, this is not true for the quantum analog (of maximum mutual information) called coherent information. In fact, the coherent information of $n$ copies of the channel can be strictly greater than $n$ times the coherent information of a single copy of the channel \cite{bennett1996mixed,divincenzo1998quantum,shor1996quantumerrorcorrectingcodesneed}. 
This lack of additivity of coherent information results in an unbounded optimization formula over the length of the input for the quantum capacity. This complicates the analysis significantly and much work has been done to remedy this situation \cite{Schumacher96,SchumacherNielsen96,lloyd1997capacity,cerf1998quantum,PhysRevA.57.4153,850671,smith2007degenerate,smith2008quantum,gsquantum, smith2011quantum, brandao2012does, cubitt2015unbounded, jackson2017degenerate, leditzky2018dephrasure, bausch2020quantum, bausch2021error, leditzky2023generic, agarwal2026error}. In spite of this, quantum capacity still remains poorly understood, even for simple channels such as the depolarizing channel.

In the quantum case, for a single copy of the channel, the coherent information can be achieved by using a random code \cite{bennett1996mixed,lloyd1997capacity,shor2002quantum,devetak2005private}\footnote{The coherent information is sometimes called the hashing bound, because random coding generalizes random hashing.}. For the depolarizing channel, a series of papers \cite{divincenzo1998quantum,shor1996quantumerrorcorrectingcodesneed,smith2007degenerate,fern} showed many examples of explicit encodings that, when concatenated with a random code, can give higher rate than a random code alone. These results also show that such encodings can achieve non-zero communication rate at higher error probabilities than those achievable by a random code. These encodings correspond to specific $n$ copy input states whose coherent information is achievable by concatenation with random coding\cite{bennett1996mixed,lloyd1997capacity,shor2002quantum,devetak2005private}.  This naturally motivates the following question: what is the best input state which gives non-zero capacity through the depolarizing channel at the highest possible error probability $p$? Finding the highest $p$ at which positive rate is possible is called the threshold problem. While previous papers established increasingly better lower bounds on $p$, there has been no improvement since the result of Fern and Whaley \cite{fern} in 2008, which gave a lower bound of $p = 0.06376$. Meanwhile, an upper bound of $p=0.08333$ can be established using a no-cloning argument \cite{bruss1998optimal}. In this paper, we give an improved lower bound on the error threshold of the depolarizing channel using input states which are maximally mixed on $2$-dimensional subspaces of the symmetric space of $n$ qubits.

\subsection{Previous Work}
In this section, we summarize the previous results for depolarizing and Pauli channels. We use the probability vector $(p_I, p_X, p_Y, p_Z)$ to denote an arbitrary mixed Pauli channel $\Phi$ with the action $\Phi(\rho) = p_I \rho + p_X X \rho X + p_Y Y \rho Y + p_Z Z \rho Z$.

Bennett, DiVincenzo, Smolin and Wootters \cite{bennett1996mixed} showed that using a random stabilizer code can attain positive rate for the $(1-3p, p, p, p)$ depolarizing channel up to $p = 0.063096$ \footnote{Note that sometimes the depolarizing channel is taken to have probabilities $(1-p,p/3,p/3,p/3)$, resulting in a difference in a factor of 3 of reported numbers.}. In a series of papers \cite{divincenzo1998quantum,shor1996quantumerrorcorrectingcodesneed,smith2007degenerate}, combinations of phase and bit-flip repetition codes (followed by a random stabilizer code) were used to improve the threshold. Specifically, Shor and Smolin \cite{shor1996quantumerrorcorrectingcodesneed} (concurrently with \cite{bennett1996mixed}) showed that by encoding with the $5$-qubit phase flip repetition code gives positive rate up to $p = 0.063452$. DiVincenzo, Shor and Smolin \cite{divincenzo1998quantum} showed that by encoding with the $5$-qubit phase flip repetition code followed by the $5$-qubit bit flip repetition code gives positive rate up to $p = 0.06352$. These repetition codes were thought to outperform random stabilizer codes because of degeneracy \cite{shor1996quantumerrorcorrectingcodesneed}. Based on this intuition, Smith and Smolin \cite{smith2007degenerate} proposed the use of ``grossly degenerate codes'' (longer length concatenated repetition codes) which map a large number of typical errors to the same syndrome, to get positive rate up to $p = 0.063626$. Soon after this, Fern and Whaley \cite{fern} gave improved lower bounds by encoding with concatenated repetition codes of significantly larger lengths, followed by $10$ layers of the $[[5,1,3]]$ code (number of qubits $\approx 5^{13}$), to get the threshold $p = 0.06376$. In this case, the threshold was estimated using a carefully crafted Monte-Carlo simulation with a runtime on the order of months. \footnote{All of these results were framed in terms of a small code concatentated with a random code.  The corresponding $n$ copy input state is the maximally mixed state on the code.}  Since their result, no further improvements on the lower bound on threshold of the depolarizing channel have been obtained.

Smith and Smolin \cite{smith2007degenerate} initiated the study of thresholds for two other Pauli channels: the independent X-Z channel with parameters $((1-p)^2, p(1-p), p^2, p(1-p))$ and the 2-Pauli channel with parameters $(1-2p, p, 0, p)$. They drew particular attention to the 2-Pauli channel since it appeared to admit no enhancement beyond random coding.  Fern and Whaley \cite{fern} used encodings similar to the ones for the depolarizing channel to give improved lower bounds on the thresholds for both these channels: $p = 0.11398$ for 2-Pauli channel and $p = 0.11277$ for independent X-Z channel. Very recently, Bhalerao and Leditzky proposed a new framework using representation theory of the symmetric group to compute the threshold for permutation invariant codes. They were able to give improved lower bounds at encoding lengths as small as $n = 18$ and $n = 24$ for the independent X-Z and 2-Pauli channels respectively.

Permutation invariant codes have also been explored previously in the context of error correction \cite{pollatsek2004permutationally,ouyang2026theory}. Very recently, permutation invariant operators were used in the context of entanglement fidelity and superactivation \cite{bergh2026permutationinvariantoptimizationproblems,parentin2026onsetsuperactivationquantumcapacity}.

\subsection{Our Results}
We give improved lower bounds on the thresholds for all the three types of Pauli channels. Our result is the first improvement on the threshold of the depolarizing channel since the result of Fern and Whaley \cite{fern} in 2008. We obtain a lower bound of $p = 0.064657$ on the threshold of the depolarizing channel at input length $n = 45$. We start getting improvements at lengths as small as $n = 24$, which is on the order of $10^{-8}$ compared to the block length required for \cite{fern}. For the independent X-Z and 2-Pauli channels we give lower bounds of $p = 0.118371$ and $p = 0.118067$ respectively on the threshold, both at input length $n = 30$. Note that we start getting improved thresholds at $n = 12$ and $n=15$, respectively.

The basis of these improvements is a generalization of the representation theoretic framework of \cite{bhalerao2025improvingquantumcommunicationrates} to include all states in the symmetric subspace, in addition to convex combinations of tensor-power states $\ket{\psi}\bra{\psi}^{\otimes n}$ considered in their paper. Lastly in Section \ref{sec:degen}, we outline a conceptual reason that explains the good performance of permutation invariant codes that we consider. In particular, we prove in Theorem \ref{thm:annihilate} and Theorem \ref{thm:rank} that for an input $\rho$ to $n$ copies of the depolarizing channel from the symmetric subspace, when we consider the CP map obtained by removing all the non-typical Pauli Kraus operators, the rank of the corresponding complementary map is exponentially smaller than the number of typical Pauli Kraus operators in the limit $n \rightarrow \infty$.  While the $n$ appearing in our optimizations are not enormous, we believe that this asymptotic counting captures  a key feature of how our symmetric codes harness degeneracy to enhance capacity thresholds.

\section{Background}
\label{sectionbackground}
\subsection{Quantum Capacity}
Given a channel $\mathcal{N}$ mapping from $A$ to $B$, and an input $\rho$, the \emph{coherent information} of $\mathcal{N}$ acting on the input $\rho$ is defined as 
\begin{align}\label{eq:channelcohinfo}
    I_c(\mathcal{N}, \rho) = S(\mathcal{N}(\rho)) - S(I \otimes \mathcal{N}(\ket{\psi}\bra{\psi}_{RA}))
\end{align}
where $\ket{\psi}_{RA}$ is a purification of $\rho$ and $S$ is the von Neumann entropy.  Note that the above expression is the same for all purifications of $\rho$.  
References \cite{lloyd1997capacity,shor2002quantum,devetak2005private} established that the above expression is an \emph{achievable rate} for the channel $\mathcal{N}$.  The number $r$ is an achievable rate if there is a family of error correcting codes, one for each block length $t$,   transmitting $rt$ qubits with error vanishing in $t$.  
The \emph{coherent information for $\mathcal{N}$} is given by 
\begin{align*}
    \mathcal{Q}^{(1)}(\mathcal{N}) = \max_\rho I_c(\mathcal{N},\rho).
\end{align*}
The superscript $1$ refers to the tensor power of $\mathcal{N}$ in the above expression. The coherent information for $\mathcal{N}$ is also called the $1$-shot coherent information for $\mathcal{N}$ or the single-letter coherent information for $\mathcal{N}$.  
The $n$-shot coherent information of $\mathcal{N}$ is 
\begin{align*}
    \mathcal{Q}^{(n)}(\mathcal{N}) = \frac{1}{n} \mathcal{Q}^{(1)}(\mathcal{N}^{\otimes n}) .
\end{align*}

The quantum capacity of $\mathcal{N}$ is defined as the supremum of all achievable rates of $\mathcal{N}$.  By \cite{lloyd1997capacity,shor2002quantum,devetak2005private} and \cite{PhysRevA.57.4153}, it is equal to the \emph{regularized} coherent information of $\mathcal{N}$,
\begin{align}
    \mathcal{Q}(\mathcal{N}) = \lim_{n\rightarrow \infty} \mathcal{Q}^{(n)}(\mathcal{N}). 
\end{align}
Following the above discussion, for any joint input $\rho$ on $A_1, \cdots, A_n$, 
the expression $\frac{1}{n} I_c(\mathcal{N}^{\otimes n},\rho)$ is an achievable rate of $\mathcal{N}$.  Since we are mostly concerned with whether the capacity is zero or positive, we can focus on the rescaled coherent information:  
\begin{align}\label{eq:quantumcapacity}
    I_c(\mathcal{N}^{\otimes n},\rho) = S(\mathcal{N}^{\otimes n}(\rho)) - S(I \otimes \mathcal{N}^{\otimes n}(\ket{\psi}\bra{\psi}_{RA_1 \cdots A_n}))
\end{align}
where $\ket{\psi}$ purifies the entire $\rho$ on a single system $R$ of appropriate dimension.  

\subsection{Representation Theory}
We start by reviewing some relevant representation theory background. See \cite{Ser97} for a more detailed exposition on representation theory. We will be interested in the following actions of the symmetric group $S_n$ and unitary group $U(d)$ on $(\mathbb{C}^d)^{\otimes n}$:
\begin{align*}
    \forall~\pi \in S_n &~~~~~~~~~~\pi \cdot (\ket{i_1, \ldots, i_n}) = U_\pi  \ket{i_1, \ldots, i_n} =\ket{i_{\pi(1)}, \ldots, i_{\pi(n)}} \\
    \forall~U \in U(d) &~~~~~~~~~~U \cdot (\ket{i_1, \ldots, i_n}) = U^{\otimes n}\ket{i_1, \ldots, i_n}
\end{align*}
We will denote the actions of $\pi \in S_n$ and $U \in U(d)$ as $U_\pi$ and $U^{\otimes n}$ respectively. The notation $U_\pi$ is chosen to reflect the fact that the action of $S_n$ is unitary. Note that the two actions commute, $M^{\otimes n} \cdot U_\pi = U_\pi \cdot M^{\otimes n}$ for all $\pi \in S_n, M \in U(d)$. The commutant of the action of $S_n$ on $(\mathbb{C}^d)^{\otimes n}$, denoted $\mathcal C(S_n)$ is the set of all operators $M \in L((\mathbb{C}^d)^{\otimes n})$ such that $M \cdot U_\pi = U_\pi \cdot M$ for all $\pi \in S_n$. This means that the unitary group $U(d)$ with its action on $(\mathbb{C}^d)^{\otimes n}$ as defined above lies in the commutant $\mathcal{C}(S_n)$ of $S_n$. Similarly, the action of $S_n$ lies in the commutant $\mathcal{C}(U(d))$ of the action of $U(d)$.

Moreover, Schur-Weyl duality \cite{fulton2013representation, etingof2011introduction} tells us that the action of $U(d)$ spans the commutant $\mathcal{C}(S_n)$ of $S_n$ and the action of $S_n$ spans the commutant $\mathcal{C}(U(d))$ of $U(d)$. In particular,
\begin{align*}
    \forall~M \in \mathcal{C}(S_n) &~~~~~~~~~~ M = \sum_{U \in U(d)} \alpha_U U^{\otimes n} \\
    \forall~M \in \mathcal{C}(U(d)) &~~~~~~~~~~ M = \sum_{\pi \in S_n} \alpha_\pi U_{\pi}
\end{align*}
Furthermore, we also have the following decomposition of $(\mathbb{C}^d)^{\otimes n}$:
\begin{align*}
    (\mathbb{C}^d)^{\otimes n} \cong \bigoplus_{\lambda\vdash_d\,n}~S_{\lambda} \otimes V_{\lambda}^d
\end{align*}
where $S_{\lambda}$ is an irrep of $S_n$ (called Specht modules) and $V^d_{\lambda}$ is an irrep of $U(d)$ (in fact, it is a polynomial irrep of $GL(d)$).  For simplicity of notation, we will drop the isomorphism $\cong$ symbol, and presume that these irreps are indeed the invariant subspaces inside $(\mathbb{C}^d)^{\otimes n}$.\footnote{In some other quantum computing literature using Schur-Weyl duality \cite{wright2016learn, harrow2005applications}, a common notation is to use the unitary operation called the Schur transform, to map an orthonormal basis for these invariant subspaces to the computational basis. On the other hand, we are mostly following the notation in \cite{bhalerao2025improvingquantumcommunicationrates}.} Since the defined actions (i.e, both $U_\pi$ and $U^{\otimes n}$) are all unitary, we further have that these invariant subspaces are orthogonal. 

Both the irreps are indexed by $\lambda$ which is a partition of $n$ into at most $d$ parts. We will represent any such partition as $\lambda = (\lambda_1, \ldots, \lambda_d)$ where $\lambda_i \geq 0$ for all $i \in [d]$, $\sum_{i=1}^d\lambda_i = n$ and $\lambda_1 \geq \ldots \lambda_d$. These partitions can be represented using a Young Diagram, we give a few examples here:
\[
\ydiagram{3,2}
\qquad
\ydiagram{5,3}
\qquad
\ydiagram{2,2}
\]
Since the actions of both $U(d)$ and $S_n$ are unitary, the subspaces isomorphic to these irreps are also orthogonal. We have the following decomposition of the action of $S_n$ and $U(d)$ on $(\mathbb{C}^d)^{\otimes n}$ as actions on the irreps:
\begin{align*}
    \forall~\pi \in S_n &~~~~~~~~~~ U_{\pi} = \bigoplus_{\lambda\vdash_d\,n}~p_{\lambda}(\pi) \otimes \mathbb{I}_{V^d_{\lambda}} \\
    \forall~U \in U(d) &~~~~~~~~~~ U^{\otimes n} = \bigoplus_{\lambda\vdash_d\,n}~\mathbb{I}_{S_{\lambda}} \otimes q^d_{\lambda}(U)
\end{align*}

i.e., on $S_\lambda$, $\pi$ acts as $p_\lambda(\pi)$ and on $V^d_\lambda$, $U$ acts as $q^d_\lambda(U)$.  In fact, Schur-Weyl duality also holds when considering $GL(d)$ with its action defined similarly to $U(d)$, and then $q_\lambda^d$ can further be extended to linear operators $L(\mathbb{C}^d)$ by continuity.

It is possible to parametrize the basis elements of $S_{\lambda}$ using Standard Young Tableaux (SYTs), and thus we can compute the dimension of $S_{\lambda}$ by counting the number of SYTs. Given a Young diagram $\lambda \vdash_d\, n$, we obtain SYTs by filling in the diagram with distinct numbers from $\{1, \ldots, n\}$ in increasing order from left to right and top to bottom, for example:
\[
\begin{ytableau}
1 & 2 & 3 & 4 & 5
\end{ytableau}
\qquad
\begin{ytableau}
1 & 2 & 3 \\
4 & 5
\end{ytableau}
\qquad
\begin{ytableau}
1 & 3 & 4 \\
2 & 5
\end{ytableau}
\]
The number of SYTs for a given $\lambda$ is given by the Hook's length formula,
\begin{align*}
    \mathsf{dim}(S_{\lambda}) = \frac{n!}{\prod_{\square \in \lambda} h_{\square}}
\end{align*}
where $h_{\square}$ is the number of boxes in the hook of $\square$, that is, the number of boxes in $\lambda$ to the right and below $\square$ (including $\square$). Using the SYTs, it is possible to obtain an orthogonal basis for $S_{\lambda}$ called the Young-Yamanouchi basis \cite{vershik2005new}.

It is possible to parametrize the basis elements of $V_{\lambda}^d$ using Semi-Standard Young Tableaux (SSYTs) with entries in $\{1,2\dots d\}$, and thus we can compute the dimension of $V^d_{\lambda}$ by counting the number of such SSYTs. Given a Young diagram $\lambda \vdash_d\, n$, we obtain these SSYTs by filling in the diagram with numbers from $\{1, \ldots, d\}$ in non-decreasing order from left to right and in increasing order from top to bottom, for example:
\[
\begin{ytableau}
1 & 1 & 1 & 2 & 2
\end{ytableau}
\qquad
\begin{ytableau}
1 & 1 & 2 \\
2 & 3
\end{ytableau}
\qquad
\begin{ytableau}
1 & 2 & 2 \\
3 & 4
\end{ytableau}
\]
The number of SSYTs with entries in $\{1,2\dots d\}$ is given by the Weyl's dimension formula \cite{christandl2006structurebipartitequantumstates},
\begin{align*}
    \mathsf{dim}(V^d_{\lambda}) = \prod_{1 \leq i < j \leq d}\frac{\lambda_i - \lambda_j + j - i}{j - i}
\end{align*}
Using the SSYTs, it is possible to obtain an orthogonal basis for $S_{\lambda}$ called the Gelfand-Tsetlin basis \cite{Mol06} which we will use to give an explicit matrix form for the irreps as well as for establishing degeneracy in Section \ref{sec:degen}.

Moreover, for any $M \in \mathcal{C}(S_n)=\mathrm{span}(U^{\otimes n})$, this decomposition tells us that 
\begin{align*}
    M = \bigoplus_{\lambda\vdash_d\,n}~\mathbb{I}_{S_{\lambda}} \otimes M_{V_\lambda},
\end{align*}
for some $M_{V_\lambda}\in L(V^d_\lambda)$ for each $\lambda\vdash_d\,n$.

We will be interested in the irreps of $U(d)$ for $d = 2$, since we will be working with channels that have qubit inputs and outputs. For this specific case, an explicit formula for the irreps $(q^2_{\lambda}, V^2_{\lambda})$ was computed in \cite{bhalerao2025improvingquantumcommunicationrates}, which we restate here. Specifically, consider an invertible $2 \times 2$ matrix $M$, and a partition $\lambda = (\lambda_1, \lambda_2)$ of $n$ where $\lambda_1 \geq \lambda_2$. Using the Weyl's dimension formula, the dimension of $V^2_{\lambda}$ is $l+1$ where $l = \lambda_1 - \lambda_2$. Then they show that with respect to the Gelfand-Tsetlin basis parametrized by the SSYTs,\footnote{Note that \cite{bhalerao2025improvingquantumcommunicationrates} parametrize the basis by Gelfand-Tsetlin patterns, which are in bijection with SSYTs.}
\begin{align*}
    q^2_{\lambda}(M) = \mathsf{det}(M)^{\lambda_2} M_{\lambda}
\end{align*}
where the $(i,j)$ entry of the $(l+1) \times (l+1)$ matrix $M_\lambda$, denoted $m_{i,j}$, is given by (Equation 3.16 of \cite{bhalerao2025improvingquantumcommunicationrates}),
\begin{align}\label{eq:irrepv}
    m_{i, j} = \sqrt{\frac{i! (l - i)!}{j!(l-j)!}}
         \sum_{k = \max\{0, i+j-l\}}^{\min\{i, j\}} \binom{j}{k}\binom{l-j}{i-k} M_{11}^k M_{12}^{j-k}M_{21}^{i-k}M_{22}^{l-i-j+k}.
\end{align}
Here, $(i, j)$ range from $l$ down to $0$.

To easily represent states in the symmetric subspace, we will use the Dicke basis. For the symmetric subspace on $n$ qubits, the Dicke basis consists of $n+1$ orthogonal elements,
\begin{align*}
    \ket{D^n_i} &= \frac{1}{\sqrt{\binom{n}{i}}} \sum_{\substack{x \in \{0,1\}^n: \\ |x| = i}} \ket{x} &\forall~i \in \{0,1, \ldots, n\}
\end{align*}
We will also use the fact that the symmetric subspace on $n$ qubits is equal to the span of $\{\ket{\phi}^{\otimes n}: \ket{\phi} \in \mathbb{C}^2\}$. In particular, we will pick $n+1$ linearly independent states $\ket{\phi_i}$ such that they span the symmetric subspace, and represent the Dicke basis as
\begin{align*}
    \ket{D_i^n} = \sum_{j=1}^{n+1} \beta_{i,j} \ket{\phi_j}^{\otimes n}
\end{align*}
This will be useful in Section \ref{sec:computingci} in order to use Equation \ref{eq:irrepv} for our computation.

\section{Computing Coherent Information}\label{sec:computingci}
 Our goal will be to compute achievable rates (thus lower bounds to the capacity) using the expression in Equation \ref{eq:quantumcapacity} for inputs from the symmetric subspace. We will do this by generalizing the method proposed in \cite{bhalerao2025improvingquantumcommunicationrates} for computing coherent information to work for all states in the symmetric subspace, in addition to the convex combinations of tensor power states considered by them. We start by stating the block-diagonal form used in Theorem 4.3 of their paper.

 \begin{lemma}[Block-Diagonalization for Tensor Power Matrices \cite{bhalerao2025improvingquantumcommunicationrates}]\label{lem:block-diag}
     Consider a $k \times k$ block-matrix $M$, where the $(i, j)$ block of $M$ is equal to $\sigma_{i, j}^{\otimes n}$. On applying a basis change corresponding to the Schur-Weyl decomposition for $(\mathbb{C}^2)^{\otimes n}$ and reordering the rows and columns, $M$ can be written in a block-diagonal form as
     \begin{align*}
         M &\cong \bigoplus_{\lambda \vdash_2\, n} \mathbb{I}_{S_{\lambda}} \otimes M_{\lambda}
     \end{align*}
     where $M_{\lambda}$ is a $k \times k$ block matrix and the $(i, j)$ block of $M_{\lambda}$ is equal to $q^2_{\lambda}(\sigma_{i,j})$.
 \end{lemma}

We will now describe how we compute the coherent information. Consider the $n$-qubit state $\rho = \frac{1}{2} (\ket{\psi_0}\bra{\psi_0} + \ket{\psi_1}\bra{\psi_1})$, where $\ket{\psi_0}$ and $\ket{\psi_1}$ are states in the symmetric subspace of $n$ qubits. We will consider the purification of $\rho$ to be $\ket{\psi} = \frac{1}{\sqrt{2}}(\ket{0}\ket{\psi_0}+ \ket{1}\ket{\psi_1})$. Note that $\ket{\psi_0}$ and $\ket{\psi_1}$ need not be orthogonal. In terms of the Dicke basis,
 \begin{align*}
     \ket{\psi_i} &= \sum_{j = 0}^{n} \alpha_{ij} \ket{D_j^n}&\text{for $i \in \{0,1\}$}
 \end{align*}
 Therefore, we have that
 \begin{align*}
     \rho &= \frac{1}{2}\sum_{i=0}^1\sum_{k, l = 0}^{n} \alpha_{i, k}\alpha_{i, l}^{\ast} \ket{D_{k}^n}\bra{D_{l}^n} \\
     \ket{\psi}\bra{\psi} &= \frac{1}{2}\sum_{i, j = 0}^1 \sum_{k, l = 0}^{n} \ket{i}\bra{j} \otimes \alpha_{i, k}\alpha_{j, l}^\ast  \ket{D_{k}^n}\bra{D_{l}^n}
 \end{align*}
By using the decomposition $\ket{D_i^n} = \sum_{j=1}^{n+1} \beta_{i,j} \ket{\phi_j}^{\otimes n}$, we get that on applying a basis change corresponding to the Schur transform,
\begin{equation}
\begin{aligned}\label{eq:diagrho}
    \ket{D_{k}^n}\bra{D_{l}^n} &\cong \bigoplus_{\lambda \vdash_2\, n} \mathbb{I}_{S_{\lambda}} \otimes \sum_{a,b = 0}^{n} \beta_{k, a} \beta_{l, b}^\ast~q^2_{\lambda}(\ket{\phi_a}\bra{\phi_b}) \\
    \rho &\cong \bigoplus_{\lambda \vdash_2\, n} \mathbb{I}_{S_{\lambda}} \otimes \frac{1}{2}\sum_{i=0}^1\sum_{k, l = 0}^{n} \alpha_{i, k}\alpha_{i, l}^{\ast} \sum_{a, b = 0}^{n} \beta_{k, a} \beta_{l, b}^\ast~q^2_{\lambda}(\ket{\phi_a}\bra{\phi_b})
\end{aligned}
\end{equation}
To get a block diagonal form for $\ket{\psi}\bra{\psi}$, we can think of $\ket{\psi}\bra{\psi}$ as a $2 \times 2$ block diagonal matrix, where the $(i,j)$ block is
\begin{align*}
    \frac{1}{2}\sum_{k, l = 0}^{n} \alpha_{i, k}\alpha_{j, l}^\ast \ket{D_{k}^n}\bra{D_{l}^n} = \frac{1}{2}\sum_{k, l = 0}^{n} \alpha_{i, k}\alpha_{j, l}^\ast \sum_{a, b = 0}^{n} \beta_{k, a} \beta_{l, b}^\ast~\ket{\phi_a}\bra{\phi_b}^{\otimes n}
\end{align*}
Note that since $\rho$ is a convex combination of states from the symmetric subspace, Equation \ref{eq:diagrho} will only have non-zero component for $\lambda = (n, 0)$, but this form will be useful for our further analysis after application of the channel. Therefore, using Lemma \ref{lem:block-diag}, we have that
\begin{align}\label{eq:diagpsi}
    \ket{\psi}\bra{\psi} \cong \bigoplus_{\lambda \vdash_2\, n} \mathbb{I}_{S_{\lambda}} \otimes \frac{1}{2} \sum_{i, j = 0}^1 \ket{i}\bra{j} \otimes \sum_{k, l = 0}^{n} \alpha_{i, k}\alpha_{j, l}^\ast \sum_{a, b = 0}^{n} \beta_{k, a} \beta_{l, b}^\ast~q^2_{\lambda}(\ket{\phi_a}\bra{\phi_b})
\end{align}
To compute the coherent information, we need to compute the entropies of $\mathcal{N}^{\otimes n}(\rho)$ and $I \otimes \mathcal{N}^{\otimes n}(\ket{\psi}\bra{\psi})$. In this case, we can use a similar analysis as Equations \ref{eq:diagrho} and \ref{eq:diagpsi} to get
\begin{align*}
    \mathcal{N}^{\otimes n}(\ket{D_{k}^n}\bra{D_{l}^n}) &\cong \bigoplus_{\lambda \vdash_2\, n} \mathbb{I}_{S_{\lambda}} \otimes \sum_{a,b = 0}^{n} \beta_{k, a} \beta_{l, b}^\ast~q^2_{\lambda}(\mathcal{N}(\ket{\phi_a}\bra{\phi_b})) \\
    \mathcal{N}^{\otimes n}(\rho) &\cong \bigoplus_{\lambda \vdash_2\, n} \mathbb{I}_{S_{\lambda}} \otimes \frac{1}{2}\sum_{i=0}^1\sum_{k, l = 0}^{n} \alpha_{i, k}\alpha_{i, l}^{\ast} \sum_{a, b = 0}^{n} \beta_{k, a} \beta_{l, b}^\ast~q^2_{\lambda}(\mathcal{N}(\ket{\phi_a}\bra{\phi_b})) \\
    I \otimes \mathcal{N}^{\otimes n}(\ket{\psi}\bra{\psi}) &\cong \bigoplus_{\lambda \vdash_2\, n} \mathbb{I}_{S_{\lambda}} \otimes \frac{1}{2} \sum_{i, j = 0}^1 \ket{i}\bra{j} \otimes \sum_{k, l = 0}^{n} \alpha_{i, k}\alpha_{j, l}^\ast \sum_{a, b = 0}^{n} \beta_{k, a} \beta_{l, b}^\ast~q^2_{\lambda}(\mathcal{N}(\ket{\phi_a}\bra{\phi_b}))
\end{align*}
We will now follow a notation similar to that of Theorem 4.3 in \cite{bhalerao2025improvingquantumcommunicationrates}, to state the formula for coherent information. Define $\sigma_{\lambda}, \omega_{\lambda}$ such that
\begin{equation}\label{eq:params}
\begin{aligned}
   \overline{q}_{\lambda} &= \frac{1}{2} \sum_{i=0}^1\sum_{k, l = 0}^{n} \alpha_{i, k}\alpha_{i, l}^{\ast} \sum_{a, b = 0}^{n} \beta_{k, a} \beta_{l, b}^\ast~\mathsf{tr(}q^2_{\lambda}(\mathcal{N}(\ket{\phi_a}\bra{\phi_b}))) \\
   c_{\lambda} &= \overline{q}_{\lambda} \cdot \mathsf{dim}(S_{\lambda}) \\
   \mathcal{N}^{\otimes n}(\rho) &\cong \bigoplus_{\lambda \vdash_2\, n} c_{\lambda} \frac{\mathbb{I}_{S_{\lambda}}}{\mathsf{dim}(S_{\lambda})} \otimes \sigma_{\lambda} \\
   I \otimes \mathcal{N}^{\otimes n}(\ket{\psi}\bra{\psi}) &\cong \bigoplus_{\lambda \vdash_2\, n} c_{\lambda} \frac{\mathbb{I}_{S_{\lambda}}}{\mathsf{dim}(S_{\lambda})} \otimes \omega_{\lambda}
\end{aligned}
\end{equation}
With this notational setup, we can state our formula for coherent information,
\begin{theorem}[Analogue of Theorem 4.3 from \cite{bhalerao2025improvingquantumcommunicationrates}]
    Given an $n$-qubit input state $\rho = \frac{1}{2}(\ket{\psi_0}\bra{\psi_0} + \ket{\psi_1}\bra{\psi_1})$ where $\ket{\psi_0}, \ket{\psi_1}$ are in the symmetric subspace of $n$-qubits. Let $\ket{\psi} = \frac{1}{\sqrt{2}}(\ket{0}\ket{\psi_0} + \ket{1}\ket{\psi_1})$ be a purification of $\rho$. The coherent information (see Equation \ref{eq:quantumcapacity}) is given by the formula
    \begin{align*}
        I_c(\mathcal{N}^{\otimes n}, \rho) &= \sum_{\lambda \vdash_2\, n} c_{\lambda} (S(\sigma_{\lambda}) - S(\omega_{\lambda}))
    \end{align*}
    where $c_{\lambda}, \sigma_{\lambda}, \omega_{\lambda}$ are as defined in Equation \ref{eq:params}.
\end{theorem}
\begin{proof}
     We will compute the entropies $S(\mathcal{N}^{\otimes n}(\rho))$ and $S(I \otimes \mathcal{N}^{\otimes n}(\ket{\psi}\bra{\psi}))$. Using Equation \ref{eq:params},
     \begin{align*}
        S(\mathcal{N}^{\otimes n}(\rho)) &= H(\mathbf{c}) + \sum_{\lambda \vdash_2\, n} c_{\lambda}(S(\sigma_{\lambda}) + \log(\mathsf{dim}(S_{\lambda}))) \\
        S(I \otimes \mathcal{N}^{\otimes n}(\ket{\psi}\bra{\psi})) &= H(\mathbf{c}) + \sum_{\lambda \vdash_2\, n} c_{\lambda}(S(\omega_{\lambda}) + \log(\mathsf{dim}(S_{\lambda})))
     \end{align*}
     where $\mathbf{c}$ is the vector formed by the coefficients $c_{\lambda}$ and $H$ is the Shannon entropy. Computing $S(\mathcal{N}^{\otimes n}(\rho)) - S(I \otimes \mathcal{N}^{\otimes n}(\ket{\psi}\bra{\psi}))$ establishes the claim.
\end{proof}
\paragraph{Code Implementation.} In order to implement this procedure, we start by picking a tensor-power basis of the form $\{\ket{\phi_i}^{\otimes n}\}_{i=1}^{n+1}$ for the symmetric subspace. This is done by picking $n+1$ single qubit states $\ket{\phi_i}$ which are ``well-spread'' on the Bloch sphere, i.e., we try to minimize the maximum overlap between each pair of states. By following this procedure, we get a low condition number for the basis change matrix between the Dicke basis and the tensor power basis, thus guaranteeing numerical stability in our results. In addition, due to numerical inaccuracy in computing the eigenvalues of density matrices (to compute the entropy), we might get negative eigenvalues. To handle that, we first project the eigenvalue vector onto the probability simplex \cite{algo}, before computing the coherent information. Given this basis, we precompute (for a given error probability $p$) the matrices $q^2_{\lambda}(\mathcal{N}(\ket{\phi_a}\bra{\phi_b}))$ for all $\lambda \vdash_2\, n$ and $a, b \in \{0, 1, \ldots, n\}$. We can then use this to precompute $\mathcal{N}^{\otimes n}(\ket{D_{k}^n}\bra{D_{l}^n})$ for all $k, l \in \{0, 1, \ldots, n\}$. Then it is possible to compute the coherent information for a given input state $\rho= \frac{1}{2}(\ket{\psi_0}\bra{\psi_0} + \ket{\psi_1}\bra{\psi_1})$ where $\ket{\psi_0}, \ket{\psi_1}$ are in the symmetric subspace of $n$-qubits. To find such states $\rho$ with high thresholds, we use a gradient descent based optimizer for coherent information. Our coherent information computation is accurate upto at least $10$ decimal places for the values of $n$ that we have computed, and the reported lower bounds on thresholds correspond to a coherent information of $10^{-7}$. The Github repository with the code can be found at \cite{github}, along with the input states used to obtain the reported thresholds.

\section{Numerical Results}
\label{sectionnumerics}
\subsection{Depolarizing Channel}
We start by discussing the results for the depolarizing channel (see Figure \ref{plot:depol}). Using our method, we obtain a lower bound of $p = 0.064657$ on the error threshold for the depolarizing channel at $n = 45$. This improves upon the prior best result of Fern and Whaley \cite{fern}, with an encoding length that is on the order of $10^{-8}$ compared to theirs. In fact, we obtain the first improved lower bound on the threshold at $n = 24$. We note that for small values of $n$, our optimization over the symmetric subspace yields the thresholds obtained by repetition codes. As such, our thresholds decrease starting from $n = 5$ to $n = 15$. However, from $n = 16$ onwards, we start seeing a rise in the thresholds (unlike repetition codes), and the threshold keeps improving as we increase $n$, till we get higher thresholds than the best ones known previously. This is in contrast to the results obtained by \cite{bhalerao2025improvingquantumcommunicationrates}, who could only recover the thresholds obtained by repetition codes by their optimization over tensor power states. While we obtain our best result at $n = 45$, our optimization has not yet yielded a peak in the threshold lower bound, i.e., even at $n = 60$ we obtain comparable thresholds, roughly around $p \approx 0.0643$. Note that our optimization does not necessarily find optimal states for a fixed $n$ as it mainly depends on gradient descent based heuristics, and the optimization under consideration is highly non-convex.

\begin{figure}
    \centering
    \includegraphics[width=0.75\linewidth]{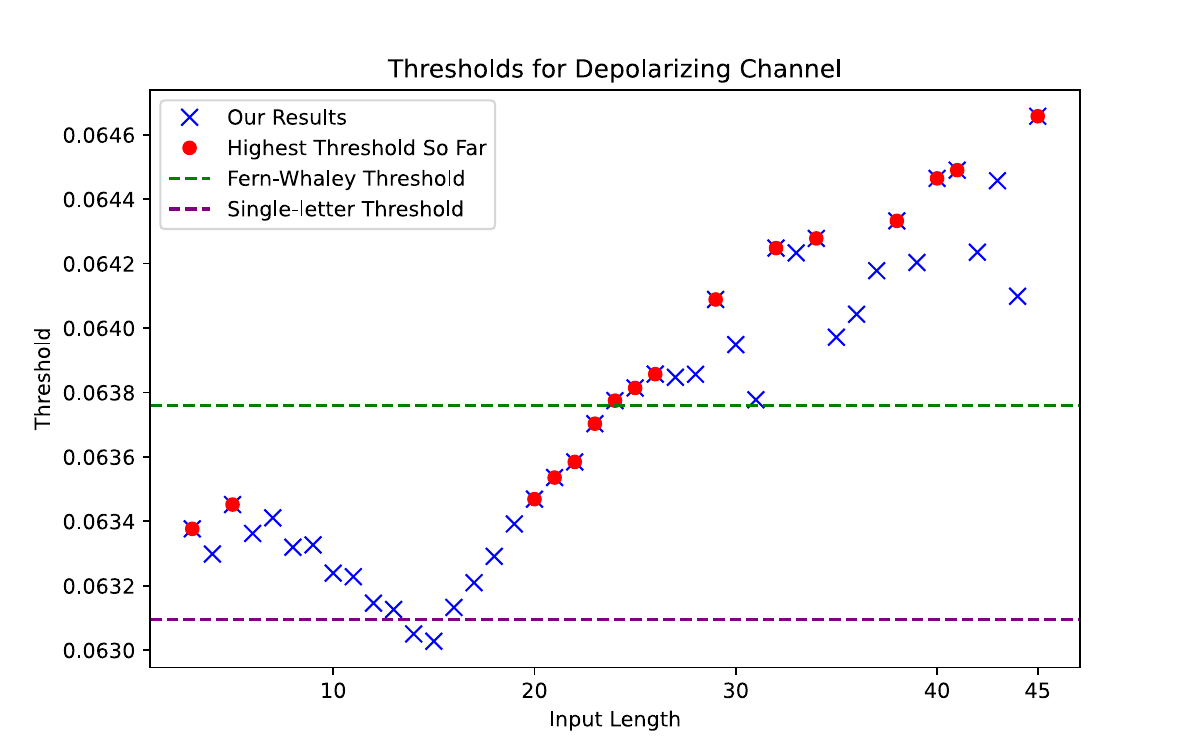}
    \caption{Error Thresholds for the Depolarizing Channel.}
    \label{fig:dep4}
    \label{plot:depol}
\end{figure}

\subsection{Independent X-Z Channel and 2-Pauli Channel}
For the independent X-Z channel (see Figure \ref{fig:indep}), we report an improved lower bound of $p = 0.118371$ on the error threshold at $n = 30$. The previous best lower bound was obtained by \cite{bhalerao2025improvingquantumcommunicationrates} at $n = 18$, and we start improving on their lower bound at a smaller encoding length of $n = 12$.

For the $2$-Pauli channel (see Figure \ref{fig:two_pauli}), we report an improved lower bound of $p = 0.118067$ on the error threshold at $n = 30$. The previous best lower bound was obtained by \cite{bhalerao2025improvingquantumcommunicationrates} at $n = 24$, and we start improving on their lower bound at a smaller encoding length of $n = 15$.

Note that for both these channels, our optimization has not yet reached its peak in terms of the lower bound on threshold, we ran our optimization and report the obtained lower bounds till $n = 30$.

In the case of $2$-Pauli channel, even obtaining lower bounds better than the single-letter threshold seems hard when using combinations of repetition codes and other stabilizer codes, and it requires concatenations of long length repetition codes (such as $5 \times 13$ or $3 \times 18$ concatenated bit and phase flip repetition codes) in order to outperform the single-letter threshold \cite{agarwal2026error}. In contrast, when optimizing over permutation invariant codes, it is possible to outperform the single-letter threshold at encoding lengths as small as $n = 6$.
\begin{figure}
    \centering
    \includegraphics[width=0.75\linewidth]{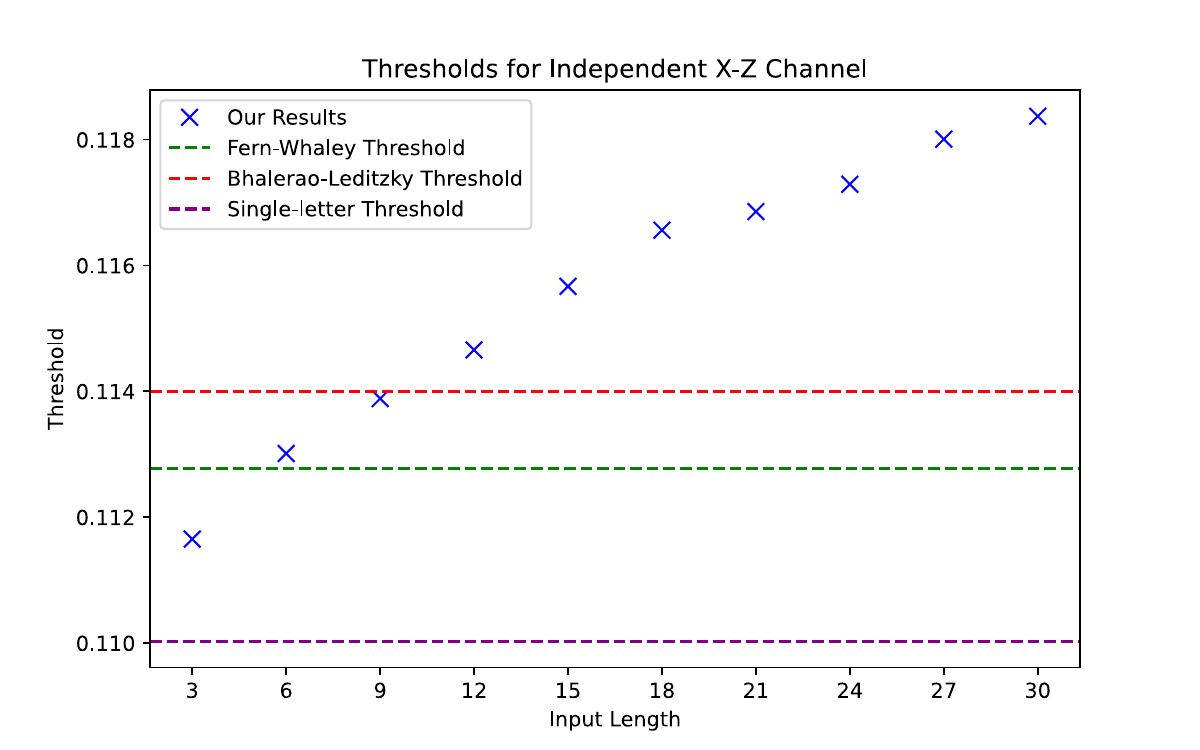}
    \caption{Error Thresholds for the Independent X-Z Channel.}
    \label{fig:indep}
\end{figure}
\begin{figure}
    \centering
    \includegraphics[width=0.75\linewidth]{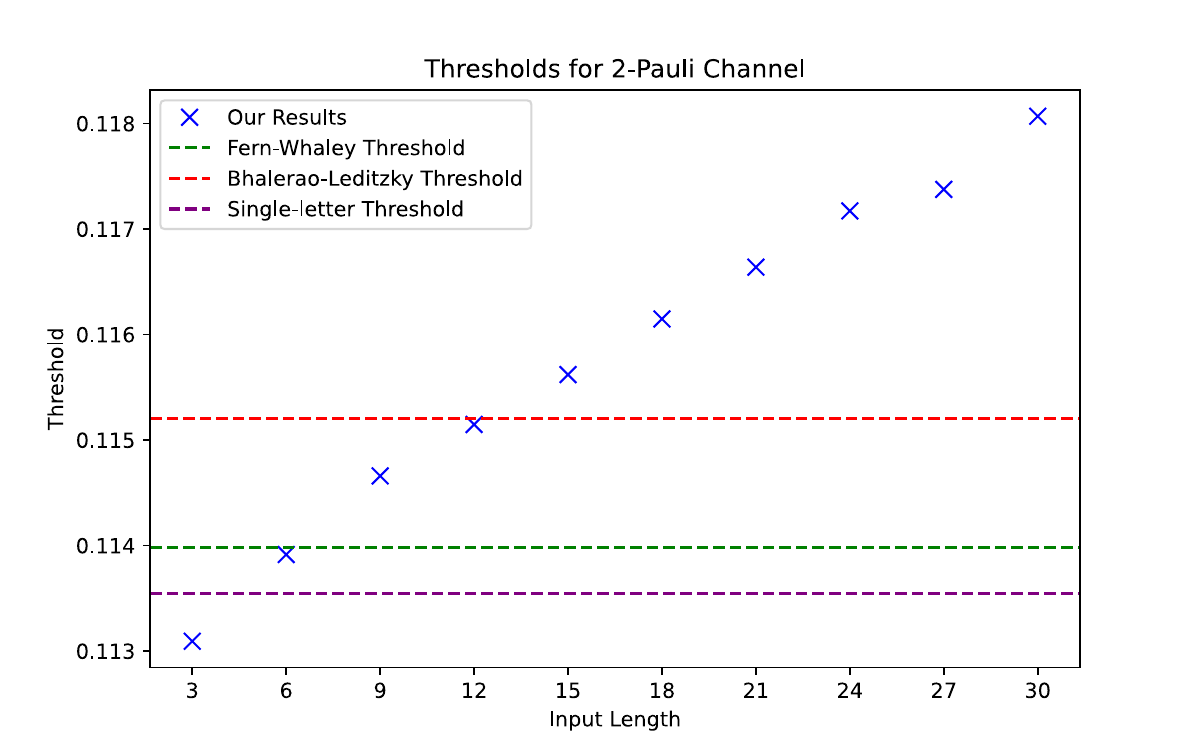}
    \caption{Error Thresholds for the 2-Pauli Channel.}
    \label{fig:two_pauli}
\end{figure}
\section{Degeneracy}\label{sec:degen}
In this section, we will tie together the good performance of permutation invariant codes to the concept of degeneracy. In \cite{shor1996quantumerrorcorrectingcodesneed,divincenzo1998quantum,smith2007degenerate} the conceptual reason for the good performance of repetition codes was that many of the typical Pauli errors are degenerate for repetition codes. Intuitively, degeneracy captures the fact that if two errors act identically on the codespace, then the entropy of the environment of the channel is reduced, since many distinct Kraus operators now collapse together with respect to this codespace. Another way to think of this is that if Kraus operators $E_1$ and $E_2$ act identically on the codespace, then $E_1 - E_2$ annihilates the codespace. We show in Theorem \ref{thm:annihilate} that when expressed in a certain representation-theoretic basis, a large number of Kraus operators of the depolarizing channel annihilate the permutation-invariant subspace.

While degeneracy is usually studied in terms of Pauli Kraus operators for stabilizer codes, here we will talk about degeneracy using Kraus operators represented in the Young-Yamanouchi basis (for $S_{\lambda}$) and the Gelfand-Tsetlin basis (for $V_{\lambda}^d$). To find the Kraus operators, we will start by finding the eigenvectors (and eigenvalues) of the Choi matrix of the channel. We will then use the permutation invariance of the input to show that many Kraus operators annihilate the input codespace. While the initial part of our discussion applies to an arbitrary channel, we will restrict ourselves to mixed Pauli channels, which means that the input and output dimensions are $2$. Given a mixed Pauli channel $\mathcal{N}$ with error probabilities $(p_I, p_X, p_Y, p_Z)$, we will represent the Choi matrix of the channel as $J$, so
\begin{align*}
    J = (I \otimes \mathcal{N})(\Phi)
\end{align*}
where $\Phi = \frac{1}{2}(\sum_{i = 0}^1 \ket{i}\ket{i})(\sum_{i = 0}^1 \bra{i}\bra{i})$. Note that $J$ is a $4 \times 4$ matrix which is diagonal in the Bell basis with diagonal entries $(p_I, p_X, p_Y, p_Z)$. Note that $J$ is a positive semi-definite matrix with trace $1$. In order to get the Kraus operators, we can find the eigenvectors $\ket{v_i}$ with eigenvalue $\gamma_i$ of $J$ such that the norm of $\ket{v_i}$ is $1$. Then we set the Kraus operators $A_i$ such that $\mathsf{vec}(A_i) = \sqrt{2\gamma_i} \ket{v_i}$. For $n$ uses of the channel, the Choi matrix is $J^{\otimes n}$. Now we can set the Kraus operators $A_i$ of $\mathcal{N}^{\otimes n}$ such that $\mathsf{vec}(A_i) = \sqrt{2^n\gamma_i} \ket{v_i}$, where $\gamma_i$ are eigenvalues and $\ket{v_i}$ are corresponding (normalized) eigenvectors. We already have a set of eigenvectors in terms of tensor products of the Bell basis, but now we will find another set of eigenvectors (and hence a different Kraus representation) using representation theory. To find these eigenvectors, we will use the following decomposition of $(\mathbb{C}^4)^{\otimes n}$:
\begin{align*}
    (\mathbb{C}^4)^{\otimes n} \cong \bigoplus_{\lambda \vdash_4\, n} S_{\lambda} \otimes V^4_{\lambda}
\end{align*}
We will use the Young-Yamanouchi basis for $S_{\lambda}$ and the Gelfand-Tsetlin basis for $V^4_{\lambda}$. We label the corresponding orthogonal basis elements of $(\mathbb{C}^4)^{\otimes n}$ as $\ket{\mathsf{SYT^i}_\lambda} \otimes \ket{\mathsf{SSYT^j}_\lambda}$ pairs for $i \in \{1, 2, \ldots, \mathsf{dim}(S_{\lambda})\}$ and $j \in \{1, 2, \ldots, \mathsf{dim}(V^4_{\lambda})\}$. Since $[J^{\otimes n}, U_{\pi}] = 0$ for all $\pi \in S_n$, it is possible to write $J^{\otimes n}$ as
\begin{align*}
    J^{\otimes n} &= \bigoplus_{\lambda\vdash_4\,n}~\mathbb{I}_{S_{\lambda}} \otimes q^4_{\lambda}(J)
\end{align*}
with respect to the basis labelling described above. In fact, we use the Gelfand-Tsetlin basis corresponding to the Bell basis, since we already knew that $J$ is diagonal in the Bell basis. This allows us to use the following property \cite{Mol06}.
\begin{lemma}[Theorem 2.3 of \cite{Mol06}]
    Let $M$ be a $d \times d$ matrix $M = \mathsf{diag}(m_1, \ldots, m_d)$ with respect to a certain basis $B = \{b_1, \ldots, b_d\}$. For $\lambda \vdash_d\, n$, consider the linear operator $q^d_{\lambda}(M)$ with respect to the Gelfand-Tsetlin basis corresponding to $B$. Then $q^d_{\lambda}(M)$ is a diagonal matrix.
\end{lemma}
Therefore, we know that $J^{\otimes n}$ is a diagonal matrix when we pick the orthogonal basis elements of $(\mathbb{C}^4)^{\otimes n}$ using the Young-Yamanouchi basis for $S_{\lambda}$ and the Gelfand-Tsetlin basis for $V^4_{\lambda}$. Therefore, we now have the eigenvectors $\ket{\mathsf{SYT}^i_\lambda} \otimes \ket{\mathsf{SSYT}^j_\lambda}$ of $J^{\otimes n}$ with respect to $(\mathbb{C}^4)^{\otimes n} \cong \bigoplus_{\lambda \vdash_4\, n} S_{\lambda} \otimes V^4_{\lambda}$, with eigenvalues $\gamma_{\lambda}^{i,j}$. Note that the eigenvalues $\gamma_{\lambda}^{i,j}$ form a probability measure and in Appendix \ref{app:techlemmas} we will refer to this as the eigenvalue measure of the Choi matrix.  In order to obtain the Kraus operators (which we will call the Schur basis Kraus operators), we want to express each of these eigenvectors using the following decomposition (in terms of input and output space):
\begin{align*}
    (\mathbb{C}^4)^{\otimes n} \cong (\mathbb{C}^2)^{\otimes n} \otimes (\mathbb{C}^2)^{\otimes n} \cong \left(\bigoplus_{\lambda \vdash_2\, n} S_{\lambda} \otimes V^2_{\lambda}\right) \otimes \left(\bigoplus_{\lambda \vdash_2\, n} S_{\lambda} \otimes V^2_{\lambda}\right)
\end{align*}
Here we will think of $(\mathbb{C}^2)^{\otimes n} \otimes (\mathbb{C}^2)^{\otimes n}$ as $\mathcal{H}^{out} \otimes \mathcal{H}^{in}$. We will label the orthogonal basis elements of $(\mathbb{C}^2)^{\otimes n}$ as $\ket{\mathsf{SYT}^i_\lambda} \otimes \ket{\mathsf{SSYT}^j_\lambda}$ pairs for $\lambda \vdash_2\, n$, $i \in \{1, 2, \ldots, \mathsf{dim}(S_{\lambda})\}$ and $j \in \{1, 2, \ldots, \mathsf{dim}(V^2_{\lambda})\}$. We will use $[k]$ to denote the set $\{1, 2, \ldots, k\}$. Therefore, for any $\lambda \vdash_4\, n$, we want to express
\begin{align}\label{eqn: kraus_in_io_form}
    \ket{\mathsf{SYT}^i_\lambda} \ket{\mathsf{SSYT}^j_\lambda} = \sum_{\lambda_a, \lambda_b \vdash_2\, n} \sum_{\substack{i_a \in [\mathsf{dim}(S_{\lambda_a})] \\i_b \in [\mathsf{dim}(S_{\lambda_b})]}}  \sum_{\substack{j_a \in [\mathsf{dim}(V^2_{\lambda_a})] \\j_b \in [\mathsf{dim}(V^2_{\lambda_b})]}}\alpha_{i, j, \lambda_a, \lambda_b}^{i_a, i_b, j_a, j_b} \ket{\mathsf{SYT}_{\lambda_a}^{i_a}} \ket{\mathsf{SSYT}_{\lambda_a}^{j_a}} \ket{\mathsf{SYT}_{\lambda_b}^{i_b}} \ket{\mathsf{SSYT}_{\lambda_b}^{j_b}}
\end{align}
We now use the permutation invariance of our input codespace. In particular, we want to project our input onto the symmetric subspace before applying the channel. Let the projector onto the symmetric subspace for the input register be $\Pi_{\mathsf{sym}}^{in}$, and the identity operator on the output subspace be $I^{out}$. Then,
\begin{align*}
    &(I^{out} \otimes \Pi_{\mathsf{sym}}^{in})\ket{\mathsf{SYT^i}_\lambda} \ket{\mathsf{SSYT^j}_\lambda} \\
    =& \sum_{\substack{\lambda_a \vdash_2\, n \\ \lambda_b = (n, 0)}} \sum_{\substack{i_a \in [\mathsf{dim}(S_{\lambda_a})] \\ i_b = 1}} \sum_{\substack{j_a \in [\mathsf{dim}(V^2_{\lambda_a})] \\ j_b \in [n+1]}} \alpha_{i, j, \lambda_a, \lambda_b}^{i_a, i_b, j_a, j_b} \ket{\mathsf{SYT}_{\lambda_a}^{i_a}} \ket{\mathsf{SSYT}_{\lambda_a}^{j_a}} \ket{\mathsf{SYT}_{\lambda_b}^{i_b}} \ket{\mathsf{SSYT}_{\lambda_b}^{j_b}}
\end{align*}
We will now use the following property (proven in Appendix \ref{app:techlemmas}) to show that for any $\lambda \vdash_4\, n$ such that $\lambda_3 > 0$, the corresponding Kraus operators annihilate the permutation invariant subspace.
\begin{restatable}{lemma}{annihilate}\label{lem:annihilate}

    Consider $\lambda \vdash_4\, n$, $i \in \{1, 2, \ldots, \mathsf{dim}(S_{\lambda})\}$ and $j \in \{1, 2, \ldots, \mathsf{dim}(V^2_{\lambda})\}$. Then, we have that
    \begin{align*}
        (I^{out} \otimes \Pi_{\mathsf{sym}}^{in})\ket{\mathsf{SYT^i}_\lambda} \ket{\mathsf{SSYT^j}_\lambda} = \sum_{\substack{i_a \in [\mathsf{dim}(S_{\lambda_a})] \\ i_b = 1}} \sum_{\substack{j_a \in [\mathsf{dim}(V^2_{\lambda_a})] \\ j_b \in [n+1]}} \alpha_{i, j, \lambda_a, \lambda_b}^{i_a, i_b, j_a, j_b} \ket{\mathsf{SYT}_{\lambda_a}^{i_a}} \ket{\mathsf{SSYT}_{\lambda_a}^{j_a}} \ket{\mathsf{SYT}_{\lambda_b}^{i_b}} \ket{\mathsf{SSYT}_{\lambda_b}^{j_b}}
    \end{align*}
    where $\lambda_a, \lambda_b \vdash_2\, n$ and $\lambda_a = \lambda, \lambda_b = (n,0)$.
\end{restatable}

Using Lemma \ref{lem:annihilate}, we get that $(I^{out} \otimes \Pi_{\mathsf{sym}}^{in})\ket{\mathsf{SYT^i}_\lambda} \ket{\mathsf{SSYT^j}_\lambda} = 0$ if $\lambda \vdash_4\, n$ such that $\lambda_3 > 0$.
\begin{theorem}\label{thm:annihilate}
    Consider a mixed Pauli channel $\mathcal{N}$ with error probabilities $(p_I, p_X, p_Y, p_Z)$, with the Choi matrix $J$. With respect to the Young-Yamanouchi basis for $S_{\lambda}$ and the Gelfand-Tsetlin basis (corresponding to the Bell basis) for $V^4_{\lambda}$, we have that
    \begin{align*}
        J^{\otimes n} &= \bigoplus_{\lambda\vdash_4\,n}~\mathbb{I}_{S_{\lambda}} \otimes q^4_{\lambda}(J)
    \end{align*}
    where $q_{\lambda}(J)$ is a diagonal matrix in the Gelfand-Tsetlin basis. On projecting the input to $n$ copies of the channel on the symmetric subspace,
    \begin{align*}
        (I^{out} \otimes \Pi_{\mathsf{sym}}^{in}) J^{\otimes n} (I^{out} \otimes \Pi_{\mathsf{sym}}^{in}) &= \bigoplus_{\lambda\vdash_2\,n}~(I^{out} \otimes \Pi_{\mathsf{sym}}^{in})(\mathbb{I}_{S_{\lambda}} \otimes q^4_{\lambda}(J))(I^{out} \otimes \Pi_{\mathsf{sym}}^{in})
    \end{align*}
    where $\Pi_{\mathsf{sym}}^{in}$ denotes the projector onto the symmetric subspace on the input register.
\end{theorem}
\begin{proof}
    We have that
    \begin{align*}
        J^{\otimes n} &= \bigoplus_{\lambda\vdash_4\,n}~\mathbb{I}_{S_{\lambda}} \otimes q^4_{\lambda}(J) \\
        &= \bigoplus_{\lambda\vdash_4\,n} \sum_{i = 1}^{\mathsf{\dim(S_{\lambda})}}\sum_{j = 1}^{\mathsf{\dim(V^4_{\lambda})}} \gamma_{\lambda}^{i, j} \ket{\mathsf{SYT^i}_\lambda} \bra{\mathsf{SYT^i}_\lambda} \otimes \ket{\mathsf{SSYT^j}_\lambda}\bra{\mathsf{SSYT^j}_\lambda}
    \end{align*}
    Therefore,
    \begin{align*}
        &(I^{out} \otimes \Pi_{\mathsf{sym}}^{in}) J^{\otimes n} (I^{out} \otimes \Pi_{\mathsf{sym}}^{in}) \\
        =& \bigoplus_{\lambda\vdash_4\,n} \sum_{i = 1}^{\mathsf{\dim(S_{\lambda})}}\sum_{j = 1}^{\mathsf{\dim(V^4_{\lambda})}} \gamma_{\lambda}^{i, j} (I^{out} \otimes \Pi_{\mathsf{sym}}^{in})(\ket{\mathsf{SYT^i}_\lambda} \bra{\mathsf{SYT^i}_\lambda} \otimes \ket{\mathsf{SSYT^j}_\lambda}\bra{\mathsf{SSYT^j}_\lambda})(I^{out} \otimes \Pi_{\mathsf{sym}}^{in}) \\
        =& \bigoplus_{\lambda\vdash_2\,n} \sum_{i = 1}^{\mathsf{\dim(S_{\lambda})}}\sum_{j = 1}^{\mathsf{\dim(V^4_{\lambda})}} \gamma_{\lambda}^{i, j} (I^{out} \otimes \Pi_{\mathsf{sym}}^{in})(\ket{\mathsf{SYT^i}_\lambda} \bra{\mathsf{SYT^i}_\lambda} \otimes \ket{\mathsf{SSYT^j}_\lambda}\bra{\mathsf{SSYT^j}_\lambda})(I^{out} \otimes \Pi_{\mathsf{sym}}^{in}) \\
        =& \bigoplus_{\lambda\vdash_2\,n}~(I^{out} \otimes \Pi_{\mathsf{sym}}^{in})(\mathbb{I}_{S_{\lambda}} \otimes q^4_{\lambda}(J))(I^{out} \otimes \Pi_{\mathsf{sym}}^{in})
    \end{align*}
    where the second equality follows from Lemma \ref{lem:annihilate}.
\end{proof}
Note that we have restricted ourselves so far to channels with qubit inputs and outputs, but the analysis of Lemma \ref{lem:annihilate} and Theorem \ref{thm:annihilate} can be generalized to channels with arbitrary input and output dimensions in a straightforward way. In the remaining part of this section, the analysis will be specialized to the depolarizing channel.

We now show for the $(1-3p, p, p, p)$ depolarizing channel that the typical Pauli Kraus operators for $\mathcal{N}^{\otimes n}$ are spanned by a subset of the Schur basis Kraus operators, such that only a small number of them do not annihilate the symmetric subspace. Let $T_{\delta}^n$ be the set of strongly $\delta$-typical Pauli errors on $n$ qubits (see Appendix \ref{app:typical}) with respect to the probability distribution $(1-3p, p, p, p)$ induced by the depolarizing channel $\mathcal{N}$ on the single-qubit Pauli errors. Let $\mathcal{N}^c$ denote the complementary channel. For an input $\sigma$ on $n$ qubits, let $\mathsf{rank}((\mathcal{N}^c)^{\otimes n}(\sigma))$ denote the rank of the output of the complementary channel. Let $\mathsf{Schur}_n$ denote the set of Schur basis Kraus operators of $\mathcal{N}^{\otimes n}$.

For $\delta > 0$, let $\mathcal N_{T^n_\delta}$ denote the CP map obtained by removing all Kraus operators from $\mathcal N^{\otimes n}$ not from the set $T_\delta^n$. Since the probability of applying an error from $T_{\delta}^n$ goes to $1$ as $n\to \infty$ (using the unit probability property from Fact \ref{fact:proptypical}), we have that $\vert\vert\mathcal N_{T^n_\delta}-\mathcal N^{\otimes n}\vert \vert_{\diamond}\to 0$ as $n \to \infty$. For this map, we correspondingly also have the complementary map $\mathcal N^c_{T^n_\delta}$ obtained by first going to the Stinespring dilation and then tracing the output register instead of the ancilla register. For this, we have the following result:

\begin{theorem}\label{thm:rank}
    For $\delta > 0$, there is a sufficiently large $n$, and a positive constant $c$, the set $T_{\delta}^n$ belongs to the span of $S \subseteq \mathsf{Schur}_n$, such that only $O\left(n^2(1-2p)^n2^{n(H(1-3p, p, p, p)+c\delta)}\right)$ elements of $S$ do not annihilate the symmetric subspace. In particular, let $\rho$ be a convex combination of states in the symmetric subspace on $n$ qubits, then
    \begin{align*}
       \mathsf{rank}((\mathcal{N}^c_{T^n_\delta})(\rho)) \in O\left(n^2(1-2p)^n2^{n(H(1-3p, p, p, p) + c\delta)}\right)
        \ll 2^{n(H(1-3p, p, p, p) + c\delta)}
        \leq \mathsf{rank}\left((\mathcal{N}^c_{T^n_\delta})\left(\frac{\mathbb{I}}{2^n}\right)\right)
    \end{align*}
\end{theorem}
\begin{proof}
    By using Lemma \ref{lem:rank} for the set $T_{\delta}^n$, and the equipartition property of $T_{\delta}^n$ from Fact \ref{fact:proptypical}, we get that $T_{\delta}^n$ is spanned by $S \subseteq \mathsf{Schur}_n$ such that the number of elements of $S$ do not annihilate the symmetric subspace is $O\left(n^2(1-2p)^n2^{n(H(1-3p, p, p, p)+c\delta)}\right)$. We will denote the non-annihilating subset of $S$ as $S'$. We now establish the second part of the claim.

    Let $A$ be the Stinespring dilation of $\mathcal{N}^{c}_{T^n_\delta}$. Let $T_{\delta}^n$ consist of the Pauli Kraus operators $\{A_k\}_{k = 1}^{|T_{\delta}^n|}$. Let $S$ consist of Kraus operators $\{B_j\}_{j = 1}^{|S|}$ where the first $|S'|$ operators belong to $S'$. Then we know from the spanning property of $S$ that
    \begin{align*}
        A_k &= \sum_{j = 1}^{|S|} a_{k,j} B_j
    \end{align*}
    Then we can think of the channel action on $\rho$ in this limit $n \rightarrow \infty, \delta \rightarrow 0$ as
    \begin{align*}
        A\rho A^{\dagger} &= \sum_{k_1, k_2 = 1}^{|T_{\delta}^n|} A_{k_1} \rho A_{k_2}^{\dagger} \otimes \ket{k_1}\bra{k_2} \\
        &= \sum_{k_1, k_2 = 1}^{|T_{\delta}^n|} \left(\sum_{j_1, j_2 = 1}^{|S|} a_{{k_1}, j_1}a_{k_2,j_2}^{\ast} B_{j_1} \rho B_{j_2}^{\dagger}\right) \otimes \ket{k_1}\bra{k_2} \\
        &= \sum_{j_1, j_2 = 1}^{|S|} B_{j_1} \rho B_{j_2}^{\dagger} \otimes \left(\sum_{k_1 = 1}^{|T_{\delta}^n|} a_{{k_1}, j_1} \ket{k_1}\right)\left(\sum_{k_2 = 1}^{|T_{\delta}^n|} a_{{k_2}, j_2}^{\ast} \bra{k_2}\right) \\
        &= \sum_{j_1, j_2 = 1}^{|S'|} B_{j_1} \rho B_{j_2}^{\dagger} \otimes \left(\sum_{k_1 = 1}^{|T_{\delta}^n|} a_{{k_1}, j_1} \ket{k_1}\right)\left(\sum_{k_2 = 1}^{|T_{\delta}^n|} a_{{k_2}, j_2}^{\ast} \bra{k_2}\right)
    \end{align*}
    where the last equality follows because only the first $|S'|$ operators in $S$ do not annihilate the symmetric subspace. Define the vectors $\ket{j} = \sum_{k = 1}^{|T_{\delta}^n|} a_{{k}, j} \ket{k}$ for $1 \leq j \leq |S'|$, and $B = \sum_{j=1}^{|S'|} B_j \otimes \ket{j}$. Then we have that
    \begin{align}\label{eq:sprimeaction}
        A\rho A^{\dagger} = B \rho B^{\dagger} = \sum_{j_1, j_2 = 1}^{|S'|} B_{j_1} \rho B_{j_2}^{\dagger} \otimes \ket{j_1}\bra{j_2}
    \end{align}
    On tracing the output register of Equation \ref{eq:sprimeaction}, the environment register is supported on the span of $\ket{j_1}\bra{j_2}$ for $1 \leq j_1, j_2 \leq |S'|$. Therefore, we get that
    \begin{align*}
        \mathsf{rank}((\mathcal{N}^c_{T^n_\delta})(\rho)) \leq |S'| \in O\left(n^2(1-2p)^n2^{n(H(1-3p, p, p, p)+c\delta)}\right) &\qedhere
    \end{align*}
\end{proof}

\section*{Acknowledgments}
AA is supported in part by a Cheriton Graduate Scholarship from the School of Computer Science at the University of Waterloo and NRC project number AQC-WAT108859. ARK acknowledges the support of the CryptoWorks 21 program and NRC project number AQC-WAT108859. AA and ARK received partial support from the NSERC Alliance Consortia Quantum grants (ALLRP 578455-22). DL is supported under RGPIN-2024-03823, LS is supported under RGPIN-2025-04875, and GS is supported under NSERC-NSF alliance grant ALLRP-586858-2023 and NSERC Discovery grant RGPIN-2025-02094. PS acknowledges the support of the Natural Sciences and Engineering
Research Council of Canada (NSERC)(RGPIN-2023-03731, ALLRP-578455-2022). Nous remercions le Conseil de recherches en sciences naturelles et en génie du Canada (CRSNG) de son soutien (RGPIN-2023-03731,ALLRP-578455-2022). PS is also supported by Institute for Quantum Computing, and the Mike and Ophelia Lazardis Fellowship. SL acknowledges the support of NSERC Alliance Horizon Europe Quantum 2021. LLMs were used to help write parts of the code to produce numerical results. The AWS cloud was also used for computing facilities.

Research at Perimeter Institute and IQC is supported by the Government of Canada through Innovation, Science and Economic Development
Canada, and by the Province of Ontario through the Ministry of Research, Innovation and Science.

\printbibliography

\appendix
\newcommand{\spec}{\mathrm{spec}}

\section{Strong Typicality}\label{app:typical}
In this section, we define and state some basic properties of strongly typical sets. See \cite{wilde2013quantum} for a detailed exposition on weak and strong typicality. Consider a set $X = \{\sigma_1, \ldots, \sigma_k\}$ and a probability distribution $p$ on $X$ which has support on all the elements of $X$. A sequence $x = (x_1, \ldots, x_n)$ where $x_i \in X$ is called a strongly $\delta$-typical sequence if
\begin{align*}
    \left|\Pr_{i \sim [n]}[x_i = \sigma_j] - p(\sigma_j)\right| &\leq \delta &\forall~~j \in [k]
\end{align*}
Clearly, if $x$ is a strongly $\delta$-typical sequence, then so are all its permutations. The set of all such strongly $\delta$-typical sequences is called the strongly $\delta$-typical set, denoted $T_{\delta}^n$. This set is then invariant under permutations. It also satisfies the following properties (see \cite{wilde2013quantum} for a proof):
\begin{fact}[Properties of strongly typical sets]\label{fact:proptypical}
    The set $T_{\delta}^n$ satisfies the following properties:
    \begin{enumerate}
        \item Unit probability: For all $\delta > 0, \epsilon \in (0,1)$, there exists $n_0 \in \mathbb{N}$ such that for all $n \geq n_0$,
        \begin{align*}
            \Pr_{x \sim X^n}[x \in T_{\delta}^n] \geq 1-\epsilon
        \end{align*}
        \item Cardinality: For all $\delta > 0, \epsilon \in (0,1)$, there exists $n_0 \in \mathbb{N}$ and a positive constant $c$ such that for all $n \geq n_0$,
        \begin{align*}
             (1-\epsilon) 2^{n(H(X)-c\delta)} \leq |T_{\delta}^n| \leq 2^{n(H(X)+c\delta)}
        \end{align*}
        \item Equipartition: For all $\delta > 0$, there exists $n_0 \in \mathbb{N}$ and a positive constant $c$ such that for all $n \geq n_0$ and $x \in T_{\delta}^n$,
        \begin{align*}
             \frac{1}{2^{n(H(X)+c\delta)}} \leq \Pr_{y \sim X^n}[y = x] \leq \frac{1}{2^{n(H(X)-c\delta)}}
        \end{align*}
    \end{enumerate}
\end{fact}

\section{Technical Lemmas}\label{app:techlemmas}
\begin{theorem}[Central Idempotents for Specht modules, Chapter 2 Theorem 8 in \cite{Ser97}]\label{thm:idempotent}
        Let $\lambda_a,\lambda_b\vdash n$.
        For any $v\in S_{\lambda_a}$
        \[C_{\lambda_b}\sum_{\pi\in S_n} \overline{\chi_{\lambda_b}(\pi)}( \pi\cdot v)=\delta_{\lambda_a,\lambda_b} v\]

        Here, $\delta_{\lambda_a,\lambda_b}$ is the Kronecker delta function, i.e, it is 1 if $\lambda_a=\lambda_b$, and $0$ otherwise. $\chi_{\lambda}(\pi)$ is the character function for $S_\lambda$ evaluated at $\pi$, which is a complex number for every $\lambda\vdash n, \pi\in S_n$. $C_{\lambda}$ is also another complex number, for each $\lambda\vdash n$.
        
    \end{theorem}
\annihilate*

\begin{proof}
    We use the central idempotent for $S_\lambda$ from Theorem \ref{thm:idempotent}. For any vector of the form $\ket{\mathsf{SYT}_{\lambda}^{i}}\otimes \ket{\mathsf{SSYT}_{\lambda}^{j}}$, we have that 
    \[C_{\lambda}\sum_{\pi\in S_n} \overline{\chi_{\lambda}(\pi)}( \pi\cdot \ket{\mathsf{SYT}_{\lambda}^{i}}\otimes \ket{\mathsf{SSYT}_{\lambda}^{j}}  )= \ket{\mathsf{SYT}_{\lambda}^{i}}\otimes \ket{\mathsf{SSYT}_{\lambda}^{j}}\]
    
    Now, when considering the isomorphism $(\mathbb C^4)^{\otimes n}\simeq (\mathbb C^2)^{\otimes n}\otimes (\mathbb C^2)^{\otimes n}$, the action of permuting the $\mathbb C^4$ tensor factors through $\pi\in S_n$ is equivalent to simultaneously permuting the $\mathbb C^2$ factors with $\pi$, for both copies of $(\mathbb C^2)^{\otimes n}$. Hence, if we rewrite vectors in $(\mathbb C^4)^{\otimes n}$ as linear combinations of vectors $v\otimes w \in (\mathbb C^2)^{\otimes n}\otimes (\mathbb C^2)^{\otimes n}$, then the action of $\pi$ can now be written as $(\pi \cdot v)\otimes (\pi\cdot w)$. On $(\mathbb C^2)^{\otimes n}$, the action of any permutation $\pi$ commutes with the projection onto the symmetric subspace. That is, for any $v\in (\mathbb C^2)^{\otimes n}$, we have $\Pi_{\mathsf{sym}} (\pi\cdot v)=\pi\cdot (\Pi_{\mathsf{sym}} v)=\Pi_{\mathsf{sym}}v$. Thus, we have that $\pi \cdot (I^{out}\otimes \Pi_{\mathsf{sym}}^{in}) = U_\pi I^{out} \otimes U_{\pi} \Pi_{\mathsf{sym}}^{in} = I^{out} U_\pi \otimes \Pi_{\mathsf{sym}}^{in} U_{\pi} = (I^{out}\otimes \Pi_{\mathsf{sym}}^{in}) (U_{\pi} \otimes U_{\pi})$. Therefore,
    \begin{align*}
        \pi \cdot ((I^{out}\otimes \Pi_{\mathsf{sym}}^{in}) \ket{\mathsf{SYT}_{\lambda}^{i}}\otimes \ket{\mathsf{SSYT}_{\lambda}^{j}}) &= (I^{out}\otimes \Pi_{\mathsf{sym}}^{in})(\pi \cdot (\ket{\mathsf{SYT}_{\lambda}^{i}}\otimes \ket{\mathsf{SSYT}_{\lambda}^{j}}))
    \end{align*}

    On the other hand, using that $S_{(n,0)}\otimes V^2_{(n,0)} $ is the symmetric subspace,
    \begin{align*}
    &(I^{out}\otimes \Pi_{\mathsf{sym}}^{in}) \ket{\mathsf{SYT}_{\lambda}^{i}}\otimes \ket{\mathsf{SSYT}_{\lambda}^{j}}\\
    = &\sum_{\substack{\lambda_a \vdash_2\, n \\ \lambda_b=(n,0)}} \sum_{\substack{i_a \in [\mathsf{dim}(S_{\lambda_a})] \\i_b \in [\mathsf{dim}(S_{\lambda_b})]}}  \sum_{\substack{j_a \in [\mathsf{dim}(V^2_{\lambda_a})] \\j_b \in [\mathsf{dim}(V^2_{\lambda_b})]}}\alpha_{i, j, \lambda_a, \lambda_b}^{i_a, i_b, j_a, j_b} \ket{\mathsf{SYT}_{\lambda_a}^{i_a}} \ket{\mathsf{SSYT}_{\lambda_a}^{j_a}} \ket{\mathsf{SYT}_{\lambda_b}^{i_b}} \ket{\mathsf{SSYT}_{\lambda_b}^{j_b}}
    \end{align*}
    Using the fact that $\pi\cdot v=v$ for any $v\in S_{(n,0)}$, we get that for all $\pi \in S_n$
    \begin{align*}
    &\pi \cdot ((I^{out}\otimes \Pi_{\mathsf{sym}}^{in}) \ket{\mathsf{SYT}_{\lambda}^{i}}\otimes \ket{\mathsf{SSYT}_{\lambda}^{j}})\\
    = &\sum_{\substack{\lambda_a \vdash_2\, n \\ \lambda_b=(n,0)}} \sum_{\substack{i_a \in [\mathsf{dim}(S_{\lambda_a})] \\i_b \in [\mathsf{dim}(S_{\lambda_b})]}}  \sum_{\substack{j_a \in [\mathsf{dim}(V^2_{\lambda_a})] \\j_b \in [\mathsf{dim}(V^2_{\lambda_b})]}}\alpha_{i, j, \lambda_a, \lambda_b}^{i_a, i_b, j_a, j_b} (\pi \cdot \ket{\mathsf{SYT}_{\lambda_a}^{i_a}}) \ket{\mathsf{SSYT}_{\lambda_a}^{j_a}} \ket{\mathsf{SYT}_{\lambda_b}^{i_b}} \ket{\mathsf{SSYT}_{\lambda_b}^{j_b}}
    \end{align*}
Then, using the central idempotent,
    \begin{align*}
    &(I^{out}\otimes \Pi_{\mathsf{sym}}^{in})\ket{\mathsf{SYT}_{\lambda}^{i}}\otimes \ket{\mathsf{SSYT}_{\lambda}^{j}} \\
    =&~(I^{out}\otimes \Pi_{\mathsf{sym}}^{in})\left(C_{\lambda}\sum_{\pi\in S_n} \overline{\chi_{\lambda}(\pi)}( \pi\cdot \ket{\mathsf{SYT}_{\lambda}^{i}}\otimes \ket{\mathsf{SSYT}_{\lambda}^{j}})\right) \\
    =&~C_{\lambda}\sum_{\pi\in S_n} \overline{\chi_{\lambda}(\pi)}( \pi\cdot((I^{out}\otimes \Pi_{\mathsf{sym}}^{in}) \ket{\mathsf{SYT}_{\lambda}^{i}}\otimes \ket{\mathsf{SSYT}_{\lambda}^{j}})  ) \\
    =& \sum_{\substack{\lambda_a \vdash_2\, n \\ \lambda_b=(n,0)}} \sum_{\substack{i_a \in [\mathsf{dim}(S_{\lambda_a})] \\i_b \in [\mathsf{dim}(S_{\lambda_b})]}}  \sum_{\substack{j_a \in [\mathsf{dim}(V^2_{\lambda_a})] \\j_b \in [\mathsf{dim}(V^2_{\lambda_b})]}}\alpha_{i, j, \lambda_a, \lambda_b}^{i_a, i_b, j_a, j_b} \left(C_{\lambda}\sum_{\pi\in S_n} \overline{\chi_{\lambda}(\pi)}(\pi\cdot \ket{\mathsf{SYT}_{\lambda_a}^{i_a}})\right)\ket{\mathsf{SSYT}_{\lambda_a}^{j_a}} \ket{\mathsf{SYT}_{\lambda_b}^{i_b}}\ket{\mathsf{SSYT}_{\lambda_b}^{j_b}}  \\
    =& \sum_{\substack{i_a \in [\mathsf{dim}(S_{\lambda_a})] \\i_b \in [\mathsf{dim}(S_{\lambda_b})]}}  \sum_{\substack{j_a \in [\mathsf{dim}(V^2_{\lambda_a})] \\j_b \in [\mathsf{dim}(V^2_{\lambda_b})]}}\alpha_{i, j, \lambda_a, \lambda_b}^{i_a, i_b, j_a, j_b} \ket{\mathsf{SYT}_{\lambda_a}^{i_a}} \ket{\mathsf{SSYT}_{\lambda_a}^{j_a}} \ket{\mathsf{SYT}_{\lambda_b}^{i_b}}\ket{\mathsf{SSYT}_{\lambda_b}^{j_b}} 
    \end{align*}
    where $\lambda_a=\lambda$ and $\lambda_b=(n,0)$ in the last line. We used that $C_{\lambda}\sum_{\pi\in S_n} \overline{\chi_{\lambda}(\pi)}(\pi\cdot \ket{\mathsf{SYT}_{\lambda_a}^{i_a}})$ is $0$ if $\lambda_a\ne \lambda$, and is $\ket{\mathsf{SYT}_{\lambda_a}^{i_a}}$ otherwise.
\end{proof}

\begin{lemma}\label{lem:two_row_prob}
    Let $\rho\in D(\mathbb C^4)$ have spectrum $(1-3p,p,p,p)$, with $p<0.2$. Consider the state $\rho^{\otimes n}$. Suppose we perform a projective measurement corresponding to orthogonal subspaces $S_\lambda\otimes V^4_{\lambda}$, which outputs the label $\lambda\vdash_4\,n$ corresponding to the irrep we measured. The probability of measuring $\lambda$ with at most $2$ rows is at most $O(n^2)\cdot (1-2p)^n$.
\end{lemma}
\begin{proof}
We have \[\rho^{\otimes n}= \sum_{\lambda \vdash_4\, l} I_{S_\lambda}\otimes q_\lambda^4(\rho).\]
The probability of measuring a fixed partition $\lambda$ is $\dim(S_\lambda)\cdot \mathsf{tr}[q_\lambda^4(\rho)]$. The second factor here is the character function for $V^4_\lambda$ evaluated at $\rho$. This is known to be equal to $s_\lambda(\spec(\rho))$, where $\spec(\circ)$ denotes the spectrum, and $s_\lambda$ is the Schur polynomial for $\lambda$. 

To work with the Schur polynomial, we will use the Jacobi's bialternant formula \cite{StanleyEC2}. For at most 4 variables, we have that

    \[s_\lambda(x_0,x_1,x_2,x_3)=\frac{\det\begin{pmatrix}
        x_0^{\lambda_1+3}&x_1^{\lambda_1+3}&x_2^{\lambda_1+3}&x_3^{\lambda_1+3}\\
        x_0^{\lambda_2+2}&x_1^{\lambda_2+2}&x_2^{\lambda_2+2}&x_3^{\lambda_2+2}\\
        x_0^{\lambda_3+1}&x_1^{\lambda_3+1}&x_2^{\lambda_3+1}&x_3^{\lambda_3+1}\\
        x_0^{\lambda_4}&x_1^{\lambda_4}&x_2^{\lambda_4}&x_3^{\lambda_4}
    \end{pmatrix}}{\det\begin{pmatrix}
        x_0^{3}&x_1^{3}&x_2^{3}&x_3^{3}\\
        x_0^{2}&x_1^{2}&x_2^{2}&x_3^{2}\\
        x_0^{1}&x_1^{1}&x_2^{1}&x_3^{1}\\
        1&1&1&1
    \end{pmatrix}}\]

    Since we are mainly interested in finding the probability of having 2 rows, we can restrict to the case $\lambda_3=\lambda_4=0$.
    With this restriction, we have

     \[s_\lambda(x_0,x_1,x_2,x_3)=\frac{\det\begin{pmatrix}
        x_0^{\lambda_1+3}&x_1^{\lambda_1+3}&x_2^{\lambda_1+3}&x_3^{\lambda_1+3}\\
        x_0^{\lambda_2+2}&x_1^{\lambda_2+2}&x_2^{\lambda_2+2}&x_3^{\lambda_2+2}\\
        x_0&x_1&x_2&x_3\\
        1&1&1&1
    \end{pmatrix}}{\det\begin{pmatrix}
        x_0^{3}&x_1^{3}&x_2^{3}&x_3^{3}\\
        x_0^{2}&x_1^{2}&x_2^{2}&x_3^{2}\\
        x_0&x_1&x_2&x_3\\
        1&1&1&1
    \end{pmatrix}}\]

Note that the above definition only works as a formal polynomial, and directly evaluating the numerator and denominator at $(x_0,x_1,x_2,x_3)=(1-3p,p,p,p)$ fails as both the numerator and denominator are 0. 

To alleviate this, we proceed as follows. For both the numerator and denominator, notice that the columns are of the form $C(x_0), C(x_1),C(x_2)$and $C(x_3)$, where $C(x)\in \mathbb C[x]^4$ (note that $C$ is different for the numerator and the denominator). Via column operations, we can transform the columns to $C(x_0),C(x_1), C(x_2)-C(x_1)$ and $C(x_3)+C(x_2)-2C(x_1)$, which doesn't change the values of the determinant. Finally, we can add a new formal variable $h$, and divide the third column by $h$, and the 4th column by $2h^2$. This leads to a scalar factor of $2h^3$ outside the determinant in both the numerator and the denominator which cancels out.

Then, for both the numerator and denominator separately, setting $x_2=x_1-h$ and $x_3=x_1+h$ and taking the limit as $h\to 0$ makes the final columns as $C(x_0),C(x_1),C'(x_1),C''(x_1)$ by using the Taylor expansion for $C(x)$ around $x_1$.

Explicitly doing it for this fraction, we get
\begin{align*}
    s_\lambda(x_0,x_1,x_1,x_1) &=\frac{\det\begin{pmatrix}
        x_0^{\lambda_1+3}&x_1^{\lambda_1+3}&(\lambda_1+3)x_1^{\lambda_1+2}&(\lambda_1+3)(\lambda_1+2)x_1^{\lambda_1+1}\\
        x_0^{\lambda_2+2}&x_1^{\lambda_2+2}&(\lambda_2+2)x_1^{\lambda_2+1}&(\lambda_2+1)(\lambda_2+2)x_1^{\lambda_2}\\
        x_0&x_1&1&0\\
        1&1&0&0
    \end{pmatrix}}{\det\begin{pmatrix}
        x_0^{3}&x_1^{3}&3x_1^{2}&6x_{1}\\
        x_0^{2}&x_1^{2}&2x_1^{1}&2\\
        x_0&x_1&1&0\\
        1&1&0&0
    \end{pmatrix}} \\
    &=\frac{\det\begin{pmatrix}
        x_0^{\lambda_1+3}&x_1^{\lambda_1+3}&(\lambda_1+3)x_1^{\lambda_1+3}&(\lambda_1+3)(\lambda_1+2)x_1^{\lambda_1+3}\\
        x_0^{\lambda_2+2}&x_1^{\lambda_2+2}&(\lambda_2+2)x_1^{\lambda_2+2}&(\lambda_2+1)(\lambda_2+2)x_1^{\lambda_2+2}\\
        x_0&x_1&x_1&0\\
        1&1&0&0
    \end{pmatrix}}{\det\begin{pmatrix}
        x_0^{3}&x_1^{3}&3x_1^{3}&6x_1^3\\
        x_0^{2}&x_1^{2}&2x_1^{2}&2x_1^2\\
        x_0&x_1&x_1&0\\
        1&1&0&0
    \end{pmatrix}},
\end{align*}
where the equality follows as we multiply powers of $x_1$ on the 3rd and 4th column for both the numerator and denominator.

We now notice that in this final form, each row has a fixed degree. Thus, both the numerator and denominator evaluate to a homogeneous polynomial in $x_0,x_1$. We now expand out the determinant corresponding to the numerator:
\begin{align*}
    &\det\begin{pmatrix}
        x_0^{\lambda_1+3}&x_1^{\lambda_1+3}&(\lambda_1+3)x_1^{\lambda_1+3}&(\lambda_1+3)(\lambda_1+2)x_1^{\lambda_1+3}\\
        x_0^{\lambda_2+2}&x_1^{\lambda_2+2}&(\lambda_2+2)x_1^{\lambda_2+2}&(\lambda_2+1)(\lambda_2+2)x_1^{\lambda_2+2}\\
        x_0&x_1&x_1&0\\
        1&1&0&0
    \end{pmatrix} \\
    =& -(\lambda_2+1)(\lambda_2+2)x_0^{\lambda_1+3}x_1^{\lambda_2+3}+(\lambda_1+3)(\lambda_1+2)x_0^{\lambda_2+2}x_1^{\lambda_1+4}\pm O(n^3)x_0x_1^{n+5} \pm O(n^3)x_1^{n+6}
\end{align*}
Thinking of this now as a polynomial function, with $x_0,x_1\geq 0$, we have that the magnitude is at most $O(n^2)x_0^{\lambda_1+2}x_1^{\lambda_2+3}(x_0+x_1) + O(n^3)x_1^{n+5}(x_0+x_1)$

    Since the denominators are common across all pairs $\lambda_1,\lambda_2$, we can sum up the numerators weighted by $\dim(S_\lambda) \leq \binom n {\lambda_1}$. From the upper bound on the magnitude of the numerator, we get that the magnitude of the weighted sum of the numerators is upper bounded by 
\begin{align*}
    &\sum_{\lambda \vdash_2\, n}\dim(S_\lambda)\left[O(n^2)x_0^{\lambda_1+2}x_1^{\lambda_2+3}(x_0+x_1)+O(n^3)x_1^{n+5}(x_0+x_1)\right] \\
    \leq &\sum_{t=0}^{n}\binom {n}{t}\left[O(n^2)x_0^{t+2}x_1^{n-t+3}(x_0+x_1)+O(n^3)x_1^{n+5}(x_0+x_1)\right] \\
    = &~O(n^2) x_0^2x_1^3(x_0+x_1)^{n+1}+O(n^3)2^{n}x_1^{n+5}(x_0+x_1)
\end{align*}
    We can also expand the denominator out as follows:
    \[\det\begin{pmatrix}
        x_0^{3}&x_1^{3}&3x_1^{3}&6x_1^3\\
        x_0^{2}&x_1^{2}&2x_1^{2}&2x_1^2\\
        x_0&x_1&x_1&0\\
        1&1&0&0
    \end{pmatrix}=-2x_1^3x_0^3+6x_0^2x_1^4-6x_0x_1^5+2x_1^6=-2(x_0-x_1)^3x_1^3\]
Plugging in $x_0=(1-3p)$ and $x_1=p$, with $p<0.2$, we get that the magnitude of the numerator is of the form $O(n^2)(1-3p)^2p^3(1-2p)^{n+1} +O(n^3)(2p)^np^{5}(1-2p)$, and the denominator is of the form $-2(1-4p)^3p^3$. Hence, we finally have the relation
\begin{align*}
    \sum_{\lambda\vdash_2\, n}\dim(S_\lambda)\cdot \Tr[q_\lambda^4(\rho)] &\leq O(n^2)(1-3p)^2(1-4p)^{-3}(1-2p)^{n+1}+ O(n^3)(2p)^{n}p^2(1-2p)(1-4p)^{-3} \\
    &\leq O(n^2)(1-2p)^{n} \qedhere
\end{align*}
\end{proof}
\begin{lemma}\label{lem:wtvector}
    
    Consider the vector  $\ket{\mathsf{SYT}^i_\lambda} \otimes \ket{\mathsf{SSYT}^j_\lambda}$, where $\ket{\mathsf{SYT}^i_\lambda}$ is a Young-Yamanouchi basis element for $S_\lambda$, and $\ket{\mathsf{SSYT}^j_\lambda}$ is a Gelfand-Tsetlin basis element for $V^4_\lambda$ corresponding to a basis $\ket{b_1},\ket{b_2},\ket{b_3},\ket{b_4}$ for $\mathbb C^4$. Then, we have that
    \begin{align*}
        \ket{\mathsf{SYT}^i_\lambda} \otimes \ket{\mathsf{SSYT}^j_\lambda} &= \sum_{\pi \in S_n} \alpha_{\pi} U_{\pi}(\ket{b_1}^{\otimes w_1} \otimes \ket{b_2}^{\otimes w_2} \otimes \ket{b_3}^{\otimes w_3} \otimes \ket{b_4}^{\otimes w_4}),
    \end{align*}
where $U_{\pi}$ is the unitary matrix corresponding to the action of permutation $\pi$. 
\end{lemma}
\begin{proof}
    Consider the matrix $M \in \mathbb{C}^4$ such that $M = \mathsf{diag}(m_1, m_2, m_3, m_4)$ in the basis $B = \{\ket{b_1}, \ket{b_2}, \ket{b_3}, \ket{b_4}\}$, where we think of $m_1, m_2, m_3, m_4$ as formal variables. Consider the eigenvector $\ket{\mathsf{SYT}^i_\lambda} \otimes \ket{\mathsf{SSYT}^j_\lambda}$ of $M^{\otimes n}$ with respect to $(\mathbb{C}^4)^{\otimes n} \cong \bigoplus_{\lambda \vdash_4\, n} S_{\lambda} \otimes V^4_{\lambda}$, using the Young-Yamanouchi basis for $S_{\lambda}$ and the Gelfand-Tsetlin basis corresponding to $B$ for $V^4_{\lambda}$.
    
    We know that that the tensor power basis constructed using $\{\ket{b_1},\ket{b_2},\ket{b_3},\ket{b_4}\}$ also form eigenvectors for $M^{\otimes n}$, with the corresponding eigenvalues being $m_1^{w_1}m_2^{w_2}m_3^{w_3}m_4^{w_4}$.

    Since the eigenvalues of $M^{\otimes n}$ remain unchanged for any basis, we know that the eigenvalue for $\ket{\mathsf{SYT}^i_\lambda} \otimes \ket{\mathsf{SSYT}^j_\lambda}$ has to be of the form $m_1^{w_1}m_2^{w_2}m_3^{w_3}m_4^{w_4}$ for $w_1, w_2, w_3, w_4 \geq 0$. We also know that the eigenspace of $M^{\otimes n}$ corresponding to this eigenvalue is equal to the span of $\{U_{\pi}(\ket{b_1}^{\otimes w_1} \otimes \ket{b_2}^{\otimes w_2} \otimes \ket{b_3}^{\otimes w_3} \otimes \ket{b_4}^{\otimes w_4})\}_{\pi \in S_n}$. Therefore, since $\ket{\mathsf{SYT}^i_\lambda} \otimes \ket{\mathsf{SSYT}^j_\lambda}$ lies in this eigenspace, the claim follows.
\end{proof}
\begin{corollary}\label{cor:probwtvector}
    For any matrix $M\in \mathbb C^4$ diagonalised in some basis $\ket{b_1},\ket{b_1},\ket{b_3},\ket{b_4}\in\mathbb C^4$, with corresponding eigenvalues $\alpha_1,\alpha_2, \alpha_3, \alpha_4$, any Schur basis element $\ket{\mathsf{SYT}^i_\lambda} \otimes \ket{\mathsf{SSYT}^j_\lambda}$ is an eigenvector for $M^{\otimes n}$ with eigenvalues of the form $\alpha_1^{w_1}\alpha_2^{w_2}\alpha_3^{w_3}\alpha_4^{w_4}$, where $\mathbf w=(w_1,w_2,w_3,w_4)$ is the unique weight vector of the tensor power basis elements appearing in the expansion of $\ket{\mathsf{SYT}^i_\lambda} \otimes \ket{\mathsf{SSYT}^j_\lambda}$.
\end{corollary}

\begin{lemma}\label{lem:rank}
Let $\mathcal N$ be the depolarizing channel corresponding to Pauli error probabilities $(1-3p,p,p,p)$. Let $T$ be a subset of Pauli errors for $\mathcal N^{\otimes n}$ which is also permutation invariant, i.e, for each Pauli operator $P \in T$, any $P'$ obtained through permuting the tensor factors of $P$ is also in $T$. Then $T$ is spanned by a subset of the Schur basis Kraus operators such that only $O(n^2(1-2p)^n p_{\min}^{-1})$ of them do not annihilate the symmetric subspace, where $p_{\min}$ is the smallest probability of a Pauli error in $T$.
\end{lemma}
\begin{proof}
    Fix a weight distribution of $I,X,Y,Z$ errors, lets say $\mathbf w=(w_I,w_X,w_Y,w_Z)$ respectively. Then, consider the set $P_{\mathbf w}$ of Pauli operators that have this weight distribution. We know that under the probability distribution of Pauli errors induced by $\mathcal N^{\otimes n}$, all elements of $P_{\mathbf w}$ are equiprobable. 

    Furthermore, the set of Schur basis Kraus operators lying in the span of $P_{\mathbf w}$ (which conversely also span $P_{\mathbf w})$ also has the exact same probabilities with respect to the measure induced by Schur basis diagonalisation of the Choi matrix for $\mathcal N^{\otimes n}$.  This follows from Corollary \ref{cor:probwtvector} and Lemma \ref{lem:wtvector}  . Call this set $K_{\mathbf w}$. Since $T$ is permutation invariant, we have a set of weight vectors $w_T$ such that $T=\cup_{\mathbf w\in w_T} P_\mathbf{w}$. Define another set $K_T=\cup_{\mathbf w\in w_T} K_{\mathbf w} $.

    Then, we have $\mathrm{span}(K_T)=\mathrm{span}(T)$ by Lemma \ref{lem:wtvector}, and $|K_T|=|T|$ by counting the dimensions (using the fact that $T$ is permutation invariant). Furthermore, with respect to the measure induced by Schur basis diagonalisation of the Choi matrix for $\mathcal N^{\otimes n}$, which is the  measure of the Choi Matrix $J^{\otimes n}$, we have that probability of sampling something from $K_T$ is equal to the probability of sampling something in $T$ under the Pauli error distribution induced by $\mathcal N^{\otimes n}$, since the vectorizations of elements from both $K_\mathbf w$ and $P_{\mathbf w}$ lie in the same eigenspace for the Choi Matrix $J^{\otimes n}$ . For simplicity, we will refer to the the probability measure over the set of Schur Basis Kraus operators as defined by the eigenvalues of $J^{\otimes n}$ as $K\sim J^{\otimes n}$. We also have that $p_{\min}$ is the smallest probability for any element of $K_T$ from Corollary \ref{cor:probwtvector}.

    Our goal now is to bound the number of elements in $K_T$ that correspond to Schur basis elements with at most $2$ rows. Assuming the number of such elements is $N_K$, we have that
    \begin{align*}
        N_K\cdot p_{\min}\leq O(n^2(1-2p)^n)
    \end{align*}
    where the upper bound is the overall probability of sampling a Schur Basis element with at most 2 rows with the distribution $K\sim J^{\otimes n}$ using Lemma \ref{lem:two_row_prob}. Thus, we have that
    \begin{align*}
        N_K\leq O(n^2(1-2p)^np_{\min}^{-1}). &\qedhere
    \end{align*}
\end{proof}

\end{document}